\documentclass[journal,letterpaper]{IEEEtran}
\IEEEoverridecommandlockouts

\usepackage{xr}%
\usepackage{calc}
\usepackage{ifthen}

\newboolean{EXT_ABS} %
\setboolean{EXT_ABS}{false} 

\usepackage[utf8]{inputenc} 
\usepackage[T1]{fontenc}
\usepackage{xurl}
\usepackage[caption=false,font=footnotesize]{subfig}
\usepackage[cmex10]{amsmath} 
\usepackage{subfiles}

\usepackage{mleftright}       
\mleftright                   

\usepackage{graphicx}         
\usepackage{booktabs}         

\usepackage{comment}
\usepackage{balance}
\usepackage[url,hyperrefblack,notheorems,IEEEtran]{research17} 
\usepackage{amsmath,amssymb,amsfonts,amsbsy}
\usepackage{tikz}
\usepackage{tikz-cd}
\usetikzlibrary{fit}
\usepackage{tablefootnote}
\usepackage[lined,boxed,commentsnumbered,linesnumbered,algoruled]{algorithm2e}
\usetikzlibrary{matrix,decorations.pathreplacing}
\pgfkeys{tikz/mymatrixenv/.style={decoration=brace,every left delimiter/.style={xshift=4pt},every right delimiter/.style={xshift=-4pt}}}
\pgfkeys{tikz/mymatrix/.style={matrix of math nodes,left delimiter=[,right delimiter={]},inner sep=1pt,row sep=0em,column sep=0em,nodes={inner sep=6pt}}}
\usepackage{amsthm} 

\newtheorem{lemma}{\mylemmaname}
\newtheorem{corollary}{\mycorollaryname}

\newtheorem{proposition}{\mypropositionname}

\newtheorem{definition}{\mydefinitionname}
\newtheorem{remark}{\myremarkname}
\newtheorem{example}{\myexamplename}
\usepackage[capitalize]{cleveref}
\allowdisplaybreaks

\usepackage{array}
\usepackage{cases}
\usepackage{multirow}
\usepackage{enumitem} 
\usepackage{nicematrix} 
\usepackage[rightcaption]{sidecap} 

\DeclareMathOperator{\maxstar}{max} 

\newcommand{\Tan}{\textnormal{Tan}}

\newcommand{\Cay}{\textnormal{Cay}}
\newcommand{\Gr}{\mathcal{G}} 

\newcommand{\G}{\mathscr{G}} 
\newcommand{\Ga}{\mathscr{G}_\textnormal{A}}
\newcommand{\Gb}{\mathscr{G}_\textnormal{B}}

\newcommand{\Gs}{\mathscr{G}^\square}

\newcommand{\vv}{\mathcal{V}}
\newcommand{\ee}{\mathcal{E}}

\newcommand{\Cc}{\mathscr{C}}
\newcommand{\Ca}{\mathscr{C}_{\textnormal{A}}}
\newcommand{\Cb}{\mathscr{C}_{\textnormal{B}}}

\newcommand{\hha}{\mat H_{\textnormal{A}}}
\newcommand{\hhb}{\mat H_{\textnormal{B}}}

\newcommand{\hhx}{\mat H_{\textnormal{X}}}
\newcommand{\hhz}{\mat H_{\textnormal{Z}}}

\newcommand{\inci}{I}  
\newcommand{\lineg}{L}  

\makeatletter
    \def\makeBB#1{\expandafter\def\csname#1bb\endcsname{\mathbb{#1}}}
    \def\makeBF#1{\expandafter\def\csname#1bf\endcsname{\mathbf{#1}}}
    \def\makeCAL#1{\expandafter\def\csname#1cal\endcsname{\mathcal{#1}}}
    \def\makeSF#1{\expandafter\def\csname#1sf\endcsname{\mathsf{#1}}}
    \def\makeSCR#1{\expandafter\def\csname#1scr\endcsname{\mathscr{#1}}}
    \def\makeTN#1{\expandafter\def\csname#1tn\endcsname{\textnormal{#1}}}

    \@for\i:={A,B,C,D,E,F,G,H,I,J,K,L,M,N,O,P,Q,R,S,T,U,V,W,X,Y,Z,a,b,c,d,e,f,g,h,i,j,k,l,m,n,o,p,q,r,s,t,u,v,w,x,y,z}\do{
        \expandafter\makeBB\i
        \expandafter\makeBF\i
        \expandafter\makeCAL\i
        \expandafter\makeSF\i
        \expandafter\makeSCR\i
        \expandafter\makeTN\i
    }
\makeatother

\usepackage{soul} 
\usepackage[colorinlistoftodos, textsize=footnotesize]{todonotes} 

\newboolean{EXTENDED_VERSION}
\setboolean{EXTENDED_VERSION}{true} 

\newboolean{SINGLE_COL} %
\setboolean{SINGLE_COL}{true} 

\newboolean{SPLIT_SUPMAT} %
\setboolean{SPLIT_SUPMAT}{true} 

\begin{document}
\title{Improved Decoding of Quantum Tanner Codes\\[0.025em] Using Generalized Check Nodes}


%
\author{Olai \AA.~Mostad,~\IEEEmembership{Student Member,~IEEE},
Eirik~Rosnes,~\IEEEmembership{Senior Member,~IEEE}, and
Hsuan-Yin~Lin,~\IEEEmembership{Senior~Member,~IEEE}
 \thanks{{Olai \AA.~Mostad, E.~Rosnes, and H.-Y.~Lin are with Simula UiB, N-5006 Bergen, Norway (e-mail: olai@simula.no, eirikrosnes@simula.no, lin@simula.no).}}}
 
\maketitle


\begin{abstract}
    We study the decoding problem for quantum Tanner codes and propose to exploit the underlying local code structure by grouping check nodes into more powerful generalized check nodes for enhanced iterative belief propagation (BP) decoding by decoding the generalized checks using a maximum a posteriori (MAP) decoder as part of the check node processing of each decoding iteration. We mainly study the finite-length setting and show that the proposed enhanced generalized BP decoder for quantum Tanner codes significantly outperforms the standard quaternary BP decoder with memory effects, as well as the recently proposed Relay-BP decoder, even outperforming generalized bicycle (GB) codes with comparable parameters in some cases. For other classes of quantum low-density parity-check  (qLDPC) codes, we propose a greedy algorithm to combine checks for generalized BP decoding. However, for GB codes, bivariate bicycle codes, hypergraph product codes, and lifted-product codes, there seems to be limited gain by combining simple checks into more powerful ones. To back up our findings, we also provide a theoretical cycle analysis for the considered qLDPC codes. This leads us to a construction criterion for quantum Tanner codes that produces codes particularly well-suited for BP decoding.
\end{abstract}





\section{Introduction}

Efficient and improved decoding of quantum low-density parity-check (qLDPC) codes has recently received significant attention, see, e.g.,~\cite{YaoLabanHagerGraellAmatPfister24_1, MiaoSchnerringLiSchmalen25_1, KuoLai22_2, Mueller-etal25_1sub, YeWeckerDelfosse25_1sub, LiuGongClark25_1sub, ValentiniForlivesiTalaricoChiani25_1sub} and references therein, as iterative belief propagation (BP) decoding (both binary and nonbinary) in general performs worse than for classical LDPC codes. The degraded performance can usually be explained by a larger number of $4$-cycles in the underlying Tanner graphs and a large number of degenerate errors,  i.e., correctable errors affecting the logical state equally as different correctable errors. 

In this work, we propose to exploit the underlying local code structure of quantum Tanner codes \cite{LeverrierZemor22_1} for improved iterative BP decoding, in the sense of exploiting a natural grouping of check nodes into \emph{generalized} checks that can be decoded using a maximum a posteriori (MAP) decoder as part of the check node processing of each decoding iteration. 
The quantum Tanner codes are constructed from a certain type of square complex by placing bits on the squares and constraints on the vertices. These constraints are more involved than single parity-checks, but can still be viewed as single generalized checks for iterative BP decoding. We also consider the intermediate case, where some checks per vertex are combined but not all.
The idea here is similar to classical generalized LDPC codes first introduced by Tanner \cite{Tanner81_1} and later extensively investigated (see, e.g.,~\cite{SimegnArtemasovAndreevRybinFrolov25_1} and references therein).  The proposed solution can be added to any iterative BP decoder for qLDPC codes, by combining checks of the parity-check matrix (preferably with meaningful overlap) into generalized checks. However, we have only observed limited gains for other classes of qLDPC codes, like generalized bicycle (GB) codes~\cite{KovalevPryadko13_1, PanteleevKalachev21_1}, bivariate bicycle (BB) codes~\cite{BravyiCrossGambettaMaslovRallYoder2024}, hypergraph product (HGP) codes~\cite{KovalevPryadko12_1, TillichZemor14_1}, and lifted-product (LP) codes~\cite{PanteleevKalachev22_2}. Here, we combine the method with quaternary (BP$_4$) decoding with additional memory effects~\cite[Alg.~1]{KuoLai22_2}, denoted by MBP$_4$, combined with ordered statistics decoding (OSD) for post-processing~\cite{PanteleevKalachev21_1}. 
Our focus is on short-to-medium length quantum Tanner codes, i.e., the so-called finite-length regime. So far, there is limited work on this regime for quantum Tanner codes, except for a few recent works \cite{MostadRosnesLin25_1, LeverrierRozendaalZemor25_1sub, RadeboldBartlettDoherty25_1sub, Wang-etal26_1sub}. 
We show by simulation that the proposed decoder can significantly outperform conventional MBP$_4$+OSD-$1$ and the recently proposed Relay-BP  (``Relay-ensembling with locally-averaged memory'') decoder \cite[Alg.~1]{Mueller-etal25_1sub} for small-to-medium length quantum Tanner codes.

As the number of $4$-cycles in the Tanner graph is important for the performance of a BP decoder, we give a description of the $4$-cycles of the Tanner graphs of quantum Tanner codes, both for the usual BP decoders and the proposed one. From this description, we derive a criterion that guarantees no $4$-cycles for a binary decoder and show how to calculate the number of $4$-cycles for a quaternary decoder.
We also provide a similar description for HGP codes, and argue that the descriptions for quantum Tanner codes and HGP codes together give such a description for LP codes.

We show that quantum Tanner codes satisfying the criterion that guarantees the absence of $4$-cycles indeed exist by giving a concrete construction for the shortest such codes, given local codes of a fixed length. Some explicit codes are constructed and are shown to perform well, even using a conventional MBP$_4$+OSD-$1$ decoder.

Our main contributions can be summarized as follows.


\begin{enumerate}
    \item Our proposed generalized MBP$_4$  decoder significantly outperforms the conventional MBP$_4$ decoder for quantum Tanner codes with and without OSD post-processing of order $1$ on the depolarization channel, at the expense of a higher  decoding complexity. Interestingly, contrary to the case of the conventional MBP$_4$ decoder, OSD post-processing is not necessary in order to approach the best performance with the  generalized MBP$_4$ decoder. However, when applied to other classes of qLDPC codes, like GB, BB, LP, and HGP codes, our simulations show no noticeable performance gain by combining simple checks into more powerful generalized checks.
         \item Building upon the previous point, we show examples where quantum Tanner codes outperform optimized GB codes, LP codes, and HGP codes with comparable parity-check row weights when decoded with the proposed generalized decoder. Moreover, as illustrated in  \cref{fig:logical_error_rate_432_different_level}, for a $[[432,16,\leq 26]]$ quantum Tanner code from \cite[Tab.~2]{Wang-etal26_1sub}, when combining only three checks into each generalized check, our proposed decoder shows  a significantly performance improvement compared to Relay-BP$_4$, even with what seems to be a lower overall decoding complexity. This example shows that the proposed decoder can also possess a favorable performance-complexity trade-off. 
     %
    \item 
    Our findings summarized above have been backed up theoretically by providing a cycle analysis for the considered qLDPC codes.

    \item 
    Finally, the cycle analysis for quantum Tanner codes leads to a criterion for constructing quantum Tanner codes with as few $4$-cycles as possible\textemdash none when using the proposed decoder. We show that these are particularly well-suited for BP decoding.

\end{enumerate}

\subsection{Notation}

Vectors are denoted by bold letters, matrices by sans serif uppercase letters, 
%
and sets (and groups)  by calligraphic uppercase letters, e.g., $\bm a$, $\mat{A}$, and $\set{A}$, respectively. The neutral element of a group will be denoted by $1$, while $e$ is reserved for an edge in a graph. Linear codes, graphs,  
and square complexes are denoted by script uppercase letters, e.g., $\mathscr{C}$.  
%
A graph with vertex set $\set{V}$ and edge set $\set{E}$ is denoted by $\code{G}=(\set{V},\set{E})$, and may have parallel edges and self-loops unless stated otherwise. The edges incident to a vertex $v$ is called the local view of $v$ and denoted $\ee(v)$. We work with undirected graphs with an ordering on the edges.
The disjoint union of sets $\set A, \set B$ is denoted by $\set A\sqcup \set B \eqdef \{(a,0), (b,1) : a\in \set A, b\in \set B\}$.
A linear code $\mathscr{C}$ of length $n$, dimension $k$, and minimum distance $d$ is sometimes referred to by 
$[n, k, d]$, and
its dual code is denoted $\mathscr{C}^\perp$. The binary field is denoted by $\Field_2$, while $\Cbb$ denotes the complex numbers. 
Standard order notation $\Theta(\cdot)$ and $O(\cdot)$ is used for asymptotic results.
{Given natural numbers $a$ and $b$, we let $[a]\eqdef\{1,2,\ldots,a\}$, $[a:b]\eqdef\{a,a+1,\ldots,b\}$, for $a \leq b$, and $[a:b]\eqdef\{a,a-1,\ldots,b\}$ for  $a > b$. For real numbers $a$ and $b$, we define $\maxstar^*(a,b) \eqdef \max(a,b) + \log(1+\mathrm{e}^{-|a-b|})$ and denote the component-wise logical OR of  $\vect a,\vect b \in \Field_2^n$ by $\vect a \lor \vect b$, that is, $(\vect a \lor \vect b)_i \eqdef a_i \lor b_i$. A quantum code with length $n$ and code dimension $k$ is denoted as $[[n,k]]$, or more explicitly as $[[n,k,d]]$ when the minimum distance $d$ is specified.
For the four Pauli matrices, we use the notation
$$
\mat I \eqdef \begin{bmatrix}
    1 & 0 \\ 0 & 1
\end{bmatrix}, 
\ \
\mat X \eqdef \begin{bmatrix}
    0 & 1 \\ 1 & 0
\end{bmatrix}, 
\ \
\mat Y \eqdef \begin{bmatrix}
    0 & -i \\ i & 0
\end{bmatrix}, 
\ \
\mat Z \eqdef \begin{bmatrix}
    1 & 0 \\ 0 & -1
\end{bmatrix}.
$$
Given $n$ Pauli matrices $\mat W_{i} \in \{\Isf, \Xsf, \Ysf, \Zsf\}$, the Kronecker product 
$\bigotimes_{i=1}^n{\mat W}_i$ is called a Pauli string, denoted by $\bm{\mat W}$, and often considered as row vectors $[\mat W_1, \dots,  \mat W_n]$.
The operator $\boxplus$ (referred to as box-plus) is defined by
$$
\mathop{\boxplus}_{i=1}^\kappa a_i \eqdef 2 \tanh^{-1} \left( \prod_{i=1}^\kappa \tanh\frac{a_i}{2}\right),
$$
for $a_1,\dots, a_\kappa \in \Rbb$.
Finally, we define $\Bar{\imath}$ as the complement of $i \in \{0,1\}$, i.e.,  $\Bar{\imath}=1$ for $i=0$ and $0$ for $i=1$, and let $\Wsf_0=\Xsf$ and  $\Wsf_1 =\Zsf$.

\section{Preliminaries}

We recall background on graphs related to error correction, the notion of quantum error-correcting codes, some particular constructions of quantum error-correcting codes, and the channel we consider for decoding.
At the reader's discretion, this section can safely be skipped and referred back to later when needed.

\subsection{Graph Constructions}\label{sec:graph-prelim}

\begin{definition}[Incidence graph]
    The (vertex-edge) incidence graph of a graph $\G = (\vv, \ee)$ is the bipartite graph $\inci(\G)$ with vertices $\vv\sqcup \ee$ and an edge $(v, e)$ for every $e\in \ee$ incident to $v\in \vv$.
\end{definition}

The incidence graph is the same as the incidence poset after we forget the direction of the edges of the incidence poset. 
A related object is the line graph.

\begin{definition}[Line graph]
    The line graph of a graph $\G = (\vv, \ee)$ is the graph $\lineg(\G)$ with vertices $\ee$ and an edge between $e, e' \in \ee$ whenever $e$ and $e'$ share a common endpoint in $\G$.
\end{definition}

In the context of quantum Tanner codes, the following definition of Tanner code is usually taken, where the restriction of a vector $\vect c \in \Field_2^{|\ee|}$ defined on the edges $\ee$ of a graph to the local view of a vertex $v$ is denoted $\vect c_v$. An ordering is assumed on the set of edges $\ee$ so that we may use $\Field_2^{|\ee|}$ instead of $\{ \ee \to \Field_2\}$ as our vector space (in which case $\vect c_v$ would be the composition $\begin{tikzcd}[column sep=1em]
    \ee(v) \subseteq \ee \ar[r, "\vect c"]
    & \Field_2
\end{tikzcd}$).

\begin{definition}[Tanner code]\label{def:Tan-old}
  Let $\mathscr{C}$ be a linear code of length $\Delta$ and $\G=(\vv,\ee)$ be a $\Delta$-regular graph, possibly with parallel edges but without self-loops. We define the Tanner code on $\G$ and $\mathscr{C}$ as
  $$\textnormal{Tan}(\G,\mathscr{C}) \triangleq \{\vect{c}\in\Field_2^{|\ee|}\colon\vect{c}_v\in\mathscr{C}\text{ for all } v\in \vv\}.$$
\end{definition}

The definition assumes that $\G$ has labeled local views. One may think of this as an order on each local view, where each order is (possibly) independent of the other orderings.
Other definitions also exist, such as Tanner's original one~\cite{Tanner81_1}. It is often demanded that the graph $\G$ is bipartite, so that one can put bits on one half of the vertices and constraints on the other half. This is essentially the same as allowing hypergraphs in \cref{def:Tan-old}, and in particular, one has
$$
\Tan(\G, \Cc) = \Tan'(\inci(\G), \Cc) = \Tan(\lineg(\inci(\G)), \Cc),
$$
where $\Tan'$ denotes the alternative definition of a Tanner code and $\lineg(\G)$ is the line graph of a graph $\G$.

    Given a parity-check matrix $\mat H$, the \textit{Tanner graph} of $\mat H$ is the bipartite graph with one check node for each row of $\mat H$, one variable node for each column of $\mat H$, and an edge between variable node $j$ and check node $i$ whenever $\mat H_{ij}=1$.
    When rows of $\mat H$ are combined into generalized checks through a collection $\Hscr = \{\Hsf_c\}_{c\in \vv_{\textnormal{c}}}$, for some set $\vv_{\textnormal{c}}$, such that
    $$
    \mat H = \left[\begin{smallmatrix}
        \Hsf_{c_1} \\ \vdots \\ \Hsf_{c_{|\vv_{\textnormal{c}}|}}
    \end{smallmatrix}\right],
    $$
    then the check nodes are $\vv_{\textnormal{c}}$. The variable nodes are the same as before, and we have an edge between $c$ and $v$ whenever column $v$ of $\mat H_c$ is nonzero. In other words, the Tanner graph for combined checks is the quotient graph of the original Tanner graph induced by the partitioning $\Hscr$ defines on the check nodes.

Our interest in studying (binary) parity-check matrices stems from their usefulness in describing quantum error correcting codes, particularly a type of code named after Calderbank, Shor, and Steane (CSS).

\subsection{Quantum Error-Correcting Codes}

\begin{definition}[CSS code]
    We say two classical linear codes $\Cscr_0, \Cscr_1$ form a CSS code $\textnormal{CSS}(\Cscr_0, \Cscr_1)$ if they satisfy $\Cscr_0^\perp \subseteq \Cscr_1$.
\end{definition}

When the classical codes $\code C_0, \code C_1$ are given by parity-check matrices $\mat H_0, \mat H_1$, respectively, the condition $\Cscr_0^\perp \subseteq \Cscr_1$ is equivalent to $\mat H_0 \trans{\mat H_1} = \mat 0$, and we sometimes write $\textnormal{CSS}(\mat H_0, \mat H_1)$ for the CSS code $\textnormal{CSS}(\code C_0, \code C_1)$. If $\Cscr_0$ and $\Cscr_1$ have codelength $n$ and dimension $k_0$ and $k_1$, respectively, then $\textnormal{CSS}(\Cscr_0, \Cscr_1)$ has length $n$ and dimension $k = k_0 + k_1 - n$. Its minimum distance is $d = \min(d_\textnormal{X}, d_\textnormal{Z})$, where
$$
d_\textnormal{X}=\min_{\vect{c}\in\mathscr{C}_0 \setminus \mathscr{C}_1^\perp}|\vect{c}|, \quad 
d_\textnormal{Z} = \min_{\vect{c}\in\mathscr{C}_1\setminus\mathscr{C}_0^\perp}|\vect{c}|.
$$

We will next give a definition for quadripartite quantum Tanner codes. There are more general code constructions that are also called quantum Tanner codes~\cite{MostadRosnesLin25_1, MostadRosnesLin25_2} for which the results of \cref{sec:cycles-QT} still hold. However, the explicit quantum Tanner codes we consider in \cref{sec:numerical-results} (see \cref{tab:codes}) all fall under the following definition.   

Let $\set A, \set B$ be subsets of a group $\Gr$ satisfying $\set A^{-1} = \set A$, $\set B^{-1} = \set B$, and $|\set A| = |\set B| = \Delta$ for some $\Delta\in \Nbb$. The quadripartite left-right Cayley complex on $\set A, \set B \subseteq \Gr$ is defined as follows. It has vertices $\vv = \Gr \times \left\{ij : i,j \in \{0,1\}\right\}$, edges $\ee = \ee_{\textnormal{A}} \cup \ee_{\textnormal{B}}$, where
\begin{align*}
    \ee_{\textnormal{A}} &= \{((g, i0),(ag, i1)): g \in \Gr, a\in \set A, i\in \{0, 1\}\} \\
    \ee_{\textnormal{B}} &= \{((g, 0j),(gb, 1j)): g \in \Gr, b\in \set B, j\in \{0, 1\}\},
\end{align*}
and squares 
$$
\{[(g, 00),(ag, 01),(gb, 10),(agb, 11)]: a \in \set A, g\in \Gr, b \in \set B\}.
$$
We write $\vv_{ij} = \Gr\times \{ij\}$, $\vv_0 = \vv_{00} \cup \vv_{11}$, and $\vv_1 = \vv_{01} \cup \vv_{10}$.
Qubits are placed on the squares, making the length of the code $n = |\vv|\Delta^2/4 = |\Gr|\Delta^2$. Restrictions will be placed on the vertices $\vv_0$ for the first part $\Cc_0$ of our CSS code, and on the vertices $\vv_1$ for $\Cc_1$. To align with \cref{def:Tan-old}, we define, for $i\in \{0, 1\}$, the graphs $\Gs_i$ with vertices $\vv_i$ and an edge for each square, connecting the two vertices of $\vv_i$ that it is incident to.

\begin{definition}[Quadripartite quantum Tanner code]\label{def:QT-code}
    Given subsets $\set A, \set B$ of a group $\Gr$ satisfying $\set A^{-1} = \set A$, $\set B^{-1} = \set B$, and classical codes $\Ca, \Cb$ of length $\Delta = |\set A| = |\set B|$ and dimensions $k_\textnormal{A} + k_\textnormal{B} = \Delta$, the quadripartite quantum Tanner code on this data is $\textnormal{CSS}(\Cc_0, \Cc_1)$, where
    \begin{align*}
        \Cc_0 &= \Tan(\Gs_0, (\Ca \otimes \Cb)^\perp) \\
        \Cc_1 &= \Tan(\Gs_1, (\Ca^\perp \otimes \Cb^\perp)^\perp).
    \end{align*}
\end{definition}

As noted already, the length of the quadripartite quantum Tanner code is $n = |\Gr|\Delta^2$. By counting checks in $\Cc_0$ and $\Cc_1$, we get the bound $k\geq |\Gr|\left(\Delta^2 - 4 k_{\textnormal{A}}k_\textnormal{B}\right)$ on the dimension.
Given parity-check matrices $\hha,\hhb, {\hha}^\perp,{\hhb}^\perp$ for the classical codes $\Ca, \Cb, {\Ca}^\perp, {\Cb}^\perp$, respectively, we can use ${\hha}^\perp\otimes {\hhb}^\perp$ and $\hha\otimes\hhb$ as parity-check matrices for $(\Ca \otimes \Cb)^\perp$ and $(\Ca^\perp \otimes \Cb^\perp)^\perp$, respectively. These can then be used to make parity-check matrices for $\Cc_0$ and $\Cc_1$ where the row weights are multiples of row weights in $\hha$ and $\hhb$ or $\hha^\perp$ and $\hhb^\perp$.

The definition of quadripartite quantum Tanner codes can be relaxed in several ways~\cite{MostadRosnesLin25_1,MostadRosnesLin25_2}. The simplest one (see \cref{tab:2-TNC-codes}), given already in~\cite{LeverrierZemor22_1}, halves the length $n$ relative to the quadripartite construction, but demands that the following condition holds.
\begin{definition}[TNC]\label{def:TNC}
    We say that subsets $\set A, \set B$ of a group $\Gr$ satisfy the Total No-Conjugacy (TNC) condition if 
    $$
    agb \neq g
    $$
    for all $a\in \set A$, $b\in \set B$, and $g\in \Gr$.
\end{definition}
When the subsets $\set A, \set B$ satisfy the TNC condition, we can use vertices $\vv = \Gr \times \{0, 1\}$, edges 
\begin{align*}
    \ee_{\textnormal{A}} &= \{((g, 0),(ag, 1)): g \in \Gr, a\in \set A \},\\
    \ee_{\textnormal{B}} &= \{((g, 0),(gb, 1)): g \in \Gr, b\in \set B\},
\end{align*}
and squares
$$
\{[(g, 0),(ag, 1),(gb, 1),(agb, 0)]: a \in \set A, g\in \Gr, b \in \set B\}
$$
for the left-right Cayley complex. Letting $\vv_0 = \Gr \times \{0\}$ and $\vv_1 = \Gr \times \{1\}$ and defining $\Gs_i$ and $\code C_i$ as before, then gives (non-quadripartite) quantum Tanner codes of length $n = \frac{1}{2}|\Gr|\Delta^2$ and dimension $k\geq \frac{1}{2}|\Gr|\left(\Delta^2 - 4 k_{\textnormal{A}}k_\textnormal{B}\right)$. The TNC condition ensures that there are no self-loops in $\Gs_i$ so that $\code C_0, \code C_1$ forms a CSS code.

HGP codes are becoming a well-studied family of qLDPC codes. They can readily be used to construct families of quantum LDPC codes with constant rate and $\Theta(\sqrt{n})$ minimum distance as the codelength $n$ tends to infinity.

\begin{definition}[HGP code]\label{def:HGP-codes}
    Given two matrices $\mat A \in \Field_2^{m_a\times n_a}$ and $\mat B \in \Field_2^{m_b\times n_b}$, the HGP code $\textnormal{HGP}(\mat A, \mat B)$ is defined as the CSS code with parity-check matrices
    $$
    \mat H_{0} \eqdef 
    \begin{bmatrix}
        \mat A\otimes \mat I_{n_b} & \mat I_{m_a} \otimes \trans{\mat B}
    \end{bmatrix},
    \quad 
    \mat H_{1} \eqdef
    \begin{bmatrix}
        \mat I_{n_a}\otimes \mat B & \trans{\mat A}\otimes \mat I_{m_b}
    \end{bmatrix}.  
    $$
\end{definition}

LP codes are lifted versions of HGP codes and are the first  qLDPC code construction proven to yield asymptotically good codes.

\begin{definition}[LP code]\label{def:LP-codes}
    Let $\hat{\mat A}\in \Field_2^{lm_a\times ln_a}$ and $\hat{\mat B} \in \Field_2^{lm_b\times ln_b}$ be block matrices consisting of $l\times l$ blocks, where the blocks of $\hat{\mat A}$ commute with the blocks of $\hat{\mat B}$. Then the LP code $\textnormal{LP}(\hat{\mat A}, \hat{\mat B})$ is the CSS code with parity-check matrices
    $$
    \mat H_{0} \eqdef 
    \begin{bmatrix}
        \hat{\mat A}\otimes \mat I_{n_b} & \mat I_{m_a} \otimes \trans{\hat{\mat B}}
    \end{bmatrix},
    \quad 
    \mat H_{1} \eqdef
    \begin{bmatrix}
        \mat I_{n_a}\otimes \hat{\mat B} & \trans{\hat{\mat A}}\otimes \mat I_{m_b}
    \end{bmatrix},  
    $$
    where the transpose is also taken inside the blocks (i.e., is the usual transpose), and $\otimes$ is done on block-level. 
\end{definition}

\subsection{Depolarizing Noise}

For our numerical results in \cref{sec:numerical-results}, we consider 
the depolarizing channel, where, for each qubit, $\Xsf$, $\Ysf$, and $\Zsf$ errors happen with equal probability $\epsilon/3$, and the identity $\Isf$ happens otherwise.

Given a CSS code $\textnormal{CSS}(\mat H_0, \mat H_1)$, we will view it as a stabilizer code by turning the rows of $\mat H_0$ into $\Xsf$-type Pauli strings and the rows of $\mat H_1$ into $\Zsf$-type Pauli strings. Specifically, we replace $1$ with $\Xsf$ and $\Zsf$, respectively, and replace $0$ with $\Isf$. These Pauli strings generate a subgroup of the $n$-fold Pauli group, which naturally acts on the state space for $n$ qubits, $(\Cbb^2)^{\otimes n}$.
The stabilizer code defined by the CSS code is the fixed points of this action by the subgroup.

We will write $\langle \cdot, \cdot \rangle$ for the function that checks whether operators commute or not. Specifically, if two Pauli strings $\bm{\Fsf}_1, \bm{\Fsf}_2$ commute, we have $\langle \bm{\Fsf}_1, \bm{\Fsf}_2\rangle = 0$, and otherwise $\langle \bm{\Fsf}_1, \bm{\Fsf}_2\rangle = 1$.
Given a set of stabilizers $\{\bm{\Fsf}_i\}$ for a code and a Pauli error $\bm{\Esf}$, we get the syndrome $\vect s$ given by $s_i = \langle \bm{\Fsf}_i, \bm{\Esf}\rangle$.
When doing error correction, we want to guess the error $\bm{\Esf}$ (or an equivalent error) based on $\vect s$. 
If some of the stabilizers are grouped together, we will write $\vect s_c$ for the part of the syndrome that comes from the stabilizers in the set $c$.





\section{Cycle Structure of Quantum Tanner Codes}\label{sec:cycles-QT}


Here, we give a description of the $4$-cycles of the Tanner graph of a quantum Tanner code, with and without combined checks. The results also hold for more general quantum Tanner codes than those of \cref{def:QT-code}, in particular the different generalizations in \cite{MostadRosnesLin25_1, MostadRosnesLin25_2}.

By recombining the checks that we put on the vertices of $\Gs_i$ in the Tanner graph of $\Cc_i$, we get a bipartite graph with one check node for each vertex of $\Gs_i$ and one variable node for each edge of $\Gs_i$, connected by an edge whenever the variable node comes from an edge adjacent to that vertex. In other words, we have the following lemma.

\begin{lemma}\label{lemma:Tanner-graph-Gsi}
    The incidence graph $\inci(\Gs_i)$ is equal to the Tanner graph of $\code C_i$ where, for each vertex in $\Gs_i$, the $k_\textnormal{A}k_\textnormal{B}$ checks associated with that vertex are combined.
\end{lemma}


When combining checks as in \cref{lemma:Tanner-graph-Gsi}, a BP$_4$ decoder will do message-passing on the combined graph $\Gs_{\textnormal{comb}} = (\vv_\textnormal{v} \sqcup \vv_\textnormal{c}, \ee)$, which can be described in the following way. 
Given a square complex $\set X$, e.g., a left-right Cayley complex, let the set of variable nodes $\vv_\textnormal{v}$ be the set of squares in $\set X$ and the set of check nodes $\vv_\textnormal{c}$ be the vertices of $\set X$. For each variable node, connect it to the four check nodes incident to that square by an edge. 
Since the square complex $\set X$ in the construction of a quantum Tanner code has a bipartition on its underlying graph (i.e., its $1$-skeleton), the check nodes are naturally split into two sets, where one is used for $\Xsf$-checks and the other for $\Zsf$-checks.


\begin{proposition}\label{prop:QT-$4$-cycles}
    $\Gs_{\textnormal{comb}}$ contains exactly these $4$-cycles:
    \begin{enumerate}
        \item\label{item:class1-cycles} The $4$-cycles associated to the parallel edges of $\Gs_0$ and $\Gs_1$,
        which for the codes of \cref{def:QT-code} correspond to $a,a'\in \set A, b,b'\in \set B, g\in \Gr$ such that 
        $$
        agb = a'gb'.
        $$
        \item\label{item:class2-cycles} The $4$-cycles from parallel edges of $\set X$.
        \item\label{item:class3-cycles} The $4$-cycles coming from squares of $\set X$ sharing at least one edge.
    \end{enumerate}
\end{proposition}
\begin{IEEEproof}
    The $4$-cycles in $\Gs_{\textnormal{comb}}$ correspond to two squares $q, p$ in the square complex $\set X$ sharing (at least) two vertices $v, w\in \vv_0 \sqcup \vv_1$. When $v, w \in \vv_0$, this corresponds to a pair of parallel edges in $\Gs_0$, and similarly for $\vv_1$ and $\Gs_1$, accounting for the cycles of \cref{item:class1-cycles}. The only other option is to have one vertex from each set, say $v \in \vv_0$ and $w \in \vv_1$. This can happen in only two ways: either $q$ and $p$ share a common edge, or they do not. If they do not share a common edge, $v$ and $w$ have a pair of parallel edges between them, 
    so the $4$-cycle is of the type of \cref{item:class2-cycles}. In the other case, they are incident to the same edge between $v$ and $w$, and we are in \cref{item:class3-cycles}.
\end{IEEEproof}

It is worth noting that each parallel edge of $\Ga$ and $\Gb$ leads to $\Delta$ parallel edges in $\Gs_0$ and $\Gs_1$, and to $2\Delta$ (assuming \cref{def:QT-code}) parallel edges in $\set X$. All parallel edges of $\set X$ come from parallel edges of $\Ga$ or $\Gb$ in this way. For the quadripartite quantum Tanner codes of \cref{def:QT-code}, all parallel edges of $\Ga$ and $\Gb$ come from repeated elements of $\set A$ and $\set B$, respectively.

\cref{prop:QT-$4$-cycles} describes the $4$-cycles in the Tanner graph where the checks coming from the same vertex of $\set X$ are combined, which our proposed decoder can use for BP decoding. When checks are not combined, the additional checks introduced by the local code must also be considered. Recall that the codes $\Ca^\perp, \Cb^\perp$ play an equivalent role to that of $\Ca, \Cb$ in the construction of a quantum Tanner code, so the local codes cannot all be good LDPC codes with Tanner graphs of large girth.

\begin{proposition}
    In the Tanner graph of $\Cscr_i$, where $\Hsf_i^{\textnormal{loc}}$ is used as a parity-check matrix for the local code, we find the following $4$-cycles.
    \begin{enumerate}
        \item[(i)] The $4$-cycles of $\Hsf_i^{\textnormal{loc}}$ appear $|\vv_i|$ times.
        \item[(ii)] Each pair of parallel edges of $\Gs_i$ labeled $(a_1,b_1)$ and $(a_2,b_2)$ in the same local view lead to 
        \begin{align*}
            &\left((\mat H_i^{\textnormal{loc}})_{a_1\Delta + b_1, :} \cdot (\mat H_i^{\textnormal{loc}})_{a_2\Delta + b_2, :}\right)
            \\
            &
            \cdot \left((\mat H_i^{\textnormal{loc}})_{a_1^{-1}\Delta + b_1^{-1}, :} \cdot (\mat H_i^{\textnormal{loc}})_{a_2^{-1}\Delta + b_2^{-1}, :}\right)
        \end{align*}
        $4$-cycles.
        %
    \end{enumerate}
    In the Tanner graph for CSS$(\Cscr_0, \Cscr_1)$, we find the following $4$-cycles between the $\Xsf$- and $\Zsf$-checks.
    \begin{enumerate}
        

        \item[(iii)] Each $4$-cycle in $\Gs_{\textnormal{comb}}$ of edges labeled $(a_1,b_1), (a_2, b_2)$ at a check-node of $\Cscr_0$ and $(a_1^{i},b_1^{-i}), (a_2^{i}, b_2^{-i})$ at a check-node of $\Cscr_1$, for $i\in \{-1, 1\}$, leads to 
        \begin{align*}
            &\left((\mat H_0^{\textnormal{loc}})_{a_1\Delta + b_1, :} \cdot (\mat H_0^{\textnormal{loc}})_{a_2\Delta + b_2, :}\right)
            \\
            &
            \cdot \left((\mat H_1^{\textnormal{loc}})_{a_1^{i}\Delta + b_1^{-i}, :} \cdot (\mat H_1^{\textnormal{loc}})_{a_2^{i}\Delta + b_2^{-i}, :}\right)
        \end{align*}
        $4$-cycles.
        Furthermore, if $i=-1$ and $a_1 = a_2$, or $i=1$ and $b_1=b_2$, this is a $4$-cycle of type $3$ of \cref{prop:QT-$4$-cycles}; otherwise it is of type $2$.
    \end{enumerate}
\end{proposition}

For ease of notation, we stick to the quadripartite construction in \cref{def:QT-code} for the following proof.

\begin{IEEEproof}
    Each variable node corresponds to a square $q$ with four vertices, as depicted below
    for the quadripartite construction:
    $$
    \begin{tikzcd}[column sep=40, row sep=30]
        (gb, 10) \ar[r, "\ a \quad a^{-1}", no head] \ar[d, swap, "b^{-1}"{near start}, "b"{near end}, no head]
        & (a g b, 11) \ar[d, "b^{-1}"{near start}, "b"{near end}, no head]
        \\
        (g, 00) \ar[r, "\ a \quad a^{-1}", no head]
        \ar[ur, "q"{description}, phantom]
        & (ag, 01)
    \end{tikzcd}
    $$
    A $4$-cycle must include two variable nodes, say coming from squares $q_1$ and $q_2$ as above, with $a,g,b$ indexed by $1$ and $2$. For the $q_1$ and $q_2$ to be part of a $4$-cycle, they must share at least one vertex. Any $4$-cycle will include two check nodes, possibly equal. For the single vertex case, each vertex of $\vv_i$ has the restriction of $\mat H_i^{\textnormal{loc}}$ placed on its local view, resulting in $|\vv_i|$ copies of the $4$-cycles of $\mat H_i^{\textnormal{loc}}$. 
    
    For the $4$-cycles including two different check nodes $c_1\in \vv_{i_1}$ and $c_2\in \vv_{i_2}$, we either have $i_1 = i_2$ or $i_1 \neq i_2$. In the first case, $q_1$ and $q_2$  both correspond to edges between $c_1$ and $c_2$ in $\Gs_{i_1}$, that is, a pair of parallel edges in $\Gs_{i_1}$. If $i_1 = i_2 = 0$, the columns of $\mat H_0$ corresponding to $q_1$ will contain the columns $(\mat H_0^{\textnormal{loc}})_{a_1\Delta + b_1, :} $
    and $(\mat H_0^{\textnormal{loc}})_{a_1^{-1}\Delta + b_1^{-1}, :}$, and otherwise be zero. The column corresponding to $q_2$ similarly contains
    $(\mat H_0^{\textnormal{loc}})_{a_2\Delta + b_2, :}$ and $(\mat H_0^{\textnormal{loc}})_{a_2^{-1}\Delta + b_2^{-1}, :}$, lining up in such a way that the inner product between the columns is precisely 
    \begin{align*}
        &\left((\mat H_0^{\textnormal{loc}})_{a_1\Delta + b_1, :} \cdot (\mat H_0^{\textnormal{loc}})_{a_2\Delta + b_2, :}\right)
        \\
        &
        + \left((\mat H_0^{\textnormal{loc}})_{a_1^{-1}\Delta + b_1^{-1}, :} \cdot (\mat H_0^{\textnormal{loc}})_{a_2^{-1}\Delta + b_2^{-1}, :}\right).
    \end{align*}
    The number of $4$-cycles between these two columns of $\mat H_0$ is the number of ways to choose two from this sum of inner products. Among these, the number of ways to choose one from the first part of the sum and one from the second part of the sum gives us the desired number of $4$-cycles.

    When the two vertices $c_1, c_2$ are not both part of $\vv_0$, the situation is the same, apart from the matrix $\mat H_1^\textnormal{loc}$ applying instead of $\mat H_0^\textnormal{loc}$ for vertices in $\vv_1$, with $a$ instead of $a^{-1}$ if $(c_1, c_2) \in \ee_\textnormal{B}$ and $b$ instead of $b^{-1}$ if $(c_1, c_2)\in \set{E}_\textnormal{A}$. 
\end{IEEEproof}


We may reasonably expect to see some improvement in decoding performance by combining only some of the checks from each vertex. In total, there are $k_{\textnormal{A}}k_{\textnormal{B}}$ checks for each vertex, so for an even split, we could combine any number dividing $k_{\textnormal{A}}k_{\textnormal{B}}$.
In particular, if $\mat H_\Atn$ and $\mat H_\Btn^\perp$ have some wanted property, then combining the $k_\Btn$ checks in $\mat H_\Atn\otimes\mat H_\Btn$ and $\mat H_\Atn^\perp\otimes\mat H_\Btn^\perp$ corresponding to the same rows of $\mat H_\Atn$ and $\mat H_\Btn^\perp$, respectively, might be a good middle ground. In that case, the proposed decoder would locally decode on $\mat H_{\textnormal{A}}\otimes [1\dots 1]$ and $[1\dots 1]\otimes \mat H_{\textnormal{B}}$. One could potentially decode the dual code locally, as well.
For ease of later reference, we sum up some of the more natural strategies for grouping checks in the following remark,
where $\Cc_{\textnormal{SPC}}$ is the single parity-check code of length $\Delta$, with parity-check matrix $[1\dots 1]$.

\begin{remark}\label{rem:ways-to-group}
    Some natural ways to group the checks of a quantum Tanner code are:
    \begin{enumerate}
        \item Combine the $k_\textnormal{A}k_\textnormal{B}$ checks coming from the same vertex of $\Gs_0$ and $\Gs_1$ (the full grouping).
        \item Combine $k_\textnormal{A}$ or $k_\textnormal{B}$ checks to obtain Tanner graphs for $\Ca^\perp\otimes \Cc_{\textnormal{SPC}}$ and $\Cc_{\textnormal{SPC}}\otimes \Cb$ or $\Ca\otimes \Cc_{\textnormal{SPC}}$ and $\Cc_{\textnormal{SPC}}\otimes \Cb^\perp$, respectively, and use this for each local view of respectively $\Gs_0$ and $\Gs_1$ (the partial grouping).
        \item Apply an algorithm on $\mat H_0^{\textnormal{loc}}$ and $\mat H_1^{\textnormal{loc}}$ for combining checks that tries to minimize the size of the link of the check nodes, and use this for each local view of $\Gs_0$ and $\Gs_1$, respectively (an unstructured local grouping).
        \item Apply an algorithm on $\mat H_0$ and $\mat H_1$ for combining checks that tries to minimize the size of the link of the check nodes (an unstructured grouping).
    \end{enumerate}
\end{remark}

We will see later that strategies 1) and 2) of \cref{rem:ways-to-group} perform well in practice, and that a fairly simple greedy algorithm (Algorithm \ref{alg:find-grouping-greedy}) in some cases can be used to beat strategy 2) using strategy 3), see \cref{sec:numerical-results} and \cref{fig:logical_error_rate_432_different_level}. We have not so far observed strategy 4) beating strategy 1).

\section{Cycle Structure of HGP and LP Codes}

We will describe (some of) the $4$-cycles of the Tanner graphs of HGP and LP codes, and how the proposed decoder can get rid of some of them. However, as shown in \cref{sec:numerical-results}, the performance gain may be limited. This is likely because the construction allows for codes with very few $4$-cycles to begin with, even none at all; see, for example, the codes C1 and C2 of~\cite{PanteleevKalachev21_1}.

Recall that HGP codes have parity-check matrices
$$
\mat H_{0} \eqdef 
\begin{bmatrix}
    \mat A\otimes \mat I_{n_b} & \mat I_{m_a} \otimes \trans{\mat B}
\end{bmatrix},
\quad 
\mat H_{1} \eqdef
\begin{bmatrix}
    \mat I_{n_a}\otimes \mat B & \trans{\mat A}\otimes \mat I_{m_b}
\end{bmatrix}.  
$$
In $\Hsf_0$, the $4$-cycles of $\Asf$ and $\Bsf$ show up $n_b$ and $m_a$ times, respectively, and in $\Hsf_1$, they show up $n_a$ and $m_b$ times, respectively. Other than that, there are no more $4$-cycles, as the inner product between columns of $\mat A\otimes \mat I_{n_b}$ and columns of $\mat I_{m_a} \otimes \trans{\mat B}$ is at most $1$, and the same holds for the blocks of $\Hsf_1$. This implies the following proposition.

\begin{proposition}
    Given matrices $\mat A \in \Field_2^{m_a\times n_a}$ and $\mat B \in \Field_2^{m_b\times n_b}$ with $t_a$ and $t_b$ $4$-cycles,  respectively, then the Tanner graphs of $\mat H_0$ and $\mat H_1$ for $\textnormal{HGP}(\mat A, \mat B)$ contain $t_a n_b + t_b m_a$ and $t_b n_a + t_a m_b$ $4$-cycles, respectively.

    Combining the checks coming from $\mat A$ ($\mat B$), the resulting Tanner graphs contain $t_b$ ($t_a$) $4$-cycles.
\end{proposition}

In a scenario where $\mat B$ has many $4$-cycles, we can get rid of them by combining checks according to collections of $n_b$ checks in $\Hsf_0$ and $m_b$ checks in $\Hsf_1$. Furthermore, if $\mat B$ is a matrix that locally has many $4$-cycles, we can get rid of them using smaller combined checks. Similarly for $\mat A$. This might allow some freedom in the construction: one matrix can be chosen among matrices with ``nice'' structure, and the other among matrices that ``fit together well'' with the first, without the same requirement on its structure. Of course, they should both preferably be LDPC and define codes with large minimum distances.

As LP codes are lifted versions of HGP codes in the same sense as quadripartite quantum Tanner codes are lifted versions of dual product codes, the above analysis can be combined with the discussion in \cref{sec:cycles-QT} to give a description of $4$-cycles in LP codes, and how our proposed decoder can deal with them. The utility of this depends on how useful the added freedom in the choice of $\mat A, \mat B$ is, at the cost of more complex decoding. In this regard, one should note that the pair $\mat A, \mat B$ probably needs to be robust in the sense of product expansion for the LP code to have a good minimum distance.


\section{Quantum Tanner Codes of Nontrivial Girth}

We provide a way to construct codes for which the Tanner graph has a larger girth when combining only a constant number of checks. The key is the following definition.

\begin{definition}[$2$-TNC]\label{def:TNC2}
    We say that subsets $\set A, \set B$ of a group $\Gr$ satisfy the $2$-step Total No-Conjugacy ($2$-TNC) condition if 
    $$
    agb \neq a'gb'
    $$
    for all $(a, b) \neq (a', b') \in \set A \times \set B$ and $g\in \Gr$.
\end{definition}

The $2$-TNC condition ensures that there are no parallel edges in $\Gs_0$ and $\Gs_1$.
Therefore, we will say that a generalized quantum Tanner code satisfies the $2$-TNC condition if there are no parallel edges in its $\Gs_0$ and $\Gs_1$.

As we saw in \cref{sec:cycles-QT}, the number of $4$-cycles in the Tanner graph of $\mat H_{0}$ and $\mat H_1$ for a quantum Tanner code relies heavily on the number of parallel edges in $\Gs_0$ and $\Gs_1$, respectively. When combining checks into generalized checks when decoding according to strategy 1) of \cref{rem:ways-to-group}, these are the only $4$-cycles we need to consider.
The following corollary is immediate from \cref{prop:QT-$4$-cycles}.

\begin{corollary}\label{cor:QT-$4$-cycles-Ci}
    If $\Gs_i$ has no parallel edges (satisfies $2$-TNC) and the parity-check matrix for the local code contains $t$ $4$-cycles, then the Tanner graph of $\Cscr_i$ contains $t|\vv_i|$ $4$-cycles.  

    When combining sets of $k_\textnormal{A}k_{\textnormal{B}}$ checks the natural way, the Tanner graph for $\Cscr_i$ has no $4$-cycles. For quadripartite quantum Tanner codes (\cref{def:QT-code}), the girth of the graph is $8$.
\end{corollary}

In \cite{MostadRosnesLin25_2}, a version of quantum Tanner codes was considered, where the product code in the definition is swapped for a Tanner code on a bipartite graph of girth $t$.
Similarly to \cref{cor:QT-$4$-cycles-Ci}, these codes can be made so that the Tanner graph of $\code C_i$ when combining checks has girth $2t$.

\begin{proposition}\label{prop:min-length-2TNC}
    If subsets $\set A, \set B$ of a group $\Gr$ satisfies the 2-TNC condition, then $|\set A| |\set B| \leq |\Gr|$.
    If the TNC condition is satisfied as well, then $|\set A| |\set B| + 1 \leq |\Gr|$.
\end{proposition}
\begin{IEEEproof}
    Consider the set of elements $agb$ where $g=1$, i.e., the set $\set{AB} \eqdef \{ab : a\in \set A, b\in \set B\} \subseteq \Gr$. Clearly $|\set{AB}| \leq |\set A||\set B|$, and 
    the inequality is strict if and only if there is some $(a,b)\neq (a',b')\in \set A \times \set B$ such that $a1b = a'1b'$, which would imply that the 2-TNC condition is not satisfied. Therefore, $|\set A||\set B| = |\set{AB}| \leq |\Gr|$ when the 2-TNC condition is satisfied.

    If the TNC condition is satisfied as well, then $a1b\neq 1$. Therefore $1\notin \set{AB}$, meaning that $|\set{AB}| \leq |\Gr| - 1$.  
\end{IEEEproof}

Recall that we require $\set A$ and $\set B$ to be closed under inverses in the quantum Tanner code construction. 

\begin{proposition}\label{prop:2TNC->TNC_sometimes}
    If subsets $\set A, \set B$ of a group $\Gr$ satisfy $\set A^{-1} = \set A$, $\set B^{-1} = \set B$, and the 2-TNC condition, then either $\set A \cap \set B = \emptyset$ or $\set A \cap \set B = \{x\}$ for some self-inverse $x\in \Gr$. 
\end{proposition}
\begin{IEEEproof}
    Assuming $|\set A \cap \set B| > 1$, let $x$ and $y$ be two different elements of $\set A \cap \set B$. Since $\set A$ and $\set B$ are closed under inverses, we also have $x^{-1}, y^{-1} \in \set A \cap \set B$. Therefore,
    $$
    x1x^{-1} = 1 = y1y^{-1}
    $$
    so that 2-TNC is not satisfied.

    If $|\set A \cap \set B| = 1$, then $x \in \set A \cap \set B$ must be self-inverse, since $\set A$ and $\set B$ are closed under inverses.
\end{IEEEproof}

If $\set A\cap \set B \neq \emptyset$, TNC can not be satisfied, since $x 1 x^{-1} = 1$. It is, however, possible to satisfy 2-TNC and not TNC even when $\set A\cap \set B \neq \emptyset$, but only if $\Gr$ is not abelian. This is clear from \cref{prop:TNC-from-2TNC-abelian}.

\begin{proposition}\label{prop:TNC-from-2TNC-abelian}
    Let $\Gr$ be an abelian group with subsets $\set A$, $\set B$ satisfying $\set A^{-1} = \set A$, $\set B^{-1} = \set B$. Then the TNC condition is satisfied if and only if $\set A \cap \set B = \emptyset$. 
\end{proposition}
\begin{IEEEproof}
    Assume the TNC condition is not satisfied, and let $a\in \set A$ and $b\in \set B$ such that $agb = g$ be given. Since $\Gr$ is abelian, we can multiply with $g^{-1}$ on both sides to get $a1b = 1$, meaning $b = a^{-1}$. Since the sets are closed under inverses, $a, b \in \set A \cap \set B$.
    Conversely, given $h\in \set A \cap \set B$, we have $h 1 h^{-1} = 1$, violating the TNC condition.
\end{IEEEproof}

By \cref{prop:min-length-2TNC}, the shortest 2-TNC codes we can construct for a given $\Delta = |\set A| = |\set B|$ have length $n = \Delta^2(\Delta + 1)/2$ and come from sets satisfying TNC. We get TNC ``for free'' by \cref{prop:TNC-from-2TNC-abelian} when $\set A$ and $\set B$ intersect trivially if the group is abelian.

\begin{example}\label{ex:short_2TNC}
    For the cyclic group of order $\Delta^2+1$, we can construct self-inverse subsets $\set A, \set B$ of order $\Delta$ satisfying both TNC and 2-TNC simultaneously.
    For simplicity, we use additive notation, so $\Gr = \Zbb_{\Delta^2 + 1}$ is the integers modulo $\Delta^2 + 1$ with neutral element $0$. 
    Let
    \begin{align*}
        \set A &= \left\{ - \frac{\Delta^2 - \Delta}{2} + (i-1) \Delta : i \in [\Delta]\right\},
        \\
        \set B &= \left\{\frac{\Delta^2 - \Delta}{2} + j : j \in [\Delta]\right\}.
    \end{align*}
    As $\set A = \left\{ - \frac{\Delta^2 - \Delta}{2} , - \frac{\Delta^2 - \Delta}{2} + \Delta, \dots , \frac{\Delta^2 - \Delta}{2} \right\}$, it is clear that the first element of $\set A$ is the inverse of the last element, and so on.
    The same is true for $\set B = \left\{\frac{\Delta^2 - \Delta}{2} + 1, \frac{\Delta^2 - \Delta}{2} + 2, \dots , \frac{\Delta^2 + \Delta}{2}\right\}$,
    so $\set A$ and $\set B$ are closed under inverses. The sum of the $i$-th element of $\set A$ and the $j$-th element of $\set B$ is 
    $$
    - \frac{\Delta^2 - \Delta}{2} + (i-1)\Delta + \frac{\Delta^2 - \Delta}{2} + j = (i-1)\Delta + j.
    $$
    Different pairs $(a,b) \in \set A \times \set B$ give different $ab\in \Gr$, so $|\set{AB}|=\Delta^2$. Since the group is abelian, 2-TNC is therefore satisfied. TNC is also satisfied by \cref{prop:TNC-from-2TNC-abelian}, as $\set A \cap \set B = \emptyset$.
    We get $\Gs_i = \set K_{\Delta^2 + 1}$, the complete graph, as we would for any $\set A, \set B \subseteq \Gr$ satisfying TNC and 2-TNC with $|\Gr| = \Delta^2+1$.
\end{example}

\cref{ex:short_2TNC} makes the following clear.
\begin{corollary}
    Given $|\set A|=|\set B|=\Delta$, the shortest quantum Tanner codes satisfying 2-TNC have length $n=\frac{1}{2}(\Delta^2+1)\Delta^2$. Such codes exist for each $\Delta\in \Nbb$.
\end{corollary}

\begin{example}\label{ex:45-7-4-2TNC}
    Consider the graphs described in \cref{ex:short_2TNC} when $\Delta=3$. Then $\Gr = \Zbb_{10}$, and we have 
    \begin{align*}
        \set A &= \left\{-\frac{3^2-3}{2}, -\frac{3^2-3}{2} + 3, -\frac{3^2-3}{2} + 6\right\} = \{7, 0, 3\}, \\
        \set B &= \left\{\frac{3^2-3}{2} + 1, \frac{3^2-3}{2} + 2, \frac{3^2-3}{2} + 3\right\} = \{ 4, 5, 6\},
    \end{align*}
    which corresponds to the following Cayley graphs, where $\Cay_\textnormal{l}(\mathbb{Z}_{10}, \set A)$ is drawn in blue and $\Cay_\textnormal{r}(\mathbb{Z}_{10}, \set B)$ is in red.
    $$
    \begin{tikzpicture}[scale=1.0,auto, swap, node distance=10m,
circ/.style={draw,shape=circle, inner sep=0.3ex},
lab/.style={shape=circle, font=\scriptsize}]
\foreach \i in {0,...,9}
{
\draw (\i*360/10: 3cm) node[circ] (\i) {$\scriptstyle \i$};
}
\foreach \i [evaluate=\i as \b using {int(Mod(\i+4,10))}, evaluate=\i as \bb using {int(Mod(\i+5,10))}, evaluate=\i as \bbb using {int(Mod(\i+6,10))}, evaluate=\i as \a using {int(Mod(\i+7,10))}, evaluate=\i as \aaa using {int(Mod(\i+3,10))}] in {0,...,9}
{
\draw[red, thick] (\i) -- (\b);
\draw[red, thick] (\i) -- (\bb);
\draw[blue, thick] (\i) -- (\a);
\draw[blue, thick] (\i) edge[out=\i * 360/10 - 380/10, in=\i * 360/10 + 380/10, loop, looseness=8] (\i);
\draw[blue, thick] (\i) -- (\aaa);
}
\end{tikzpicture}
    $$
    The TNC condition corresponds to no overlap between red and blue edges in this drawing.
    The graphs $\Gs_0$ and $\Gs_1$, where edges can be viewed as compositions of first a blue and then a red edge (or the other way around), are now both the complete graph on 10 vertices, $\set K_{\Delta^2 + 1}$, drawn below. The 2-TNC condition holds if and only if there are no parallel edges in this graph.
    $$
    \begin{tikzpicture}[scale=1.0,auto, swap, node distance=10m,
circ/.style={draw,shape=circle, inner sep=0.3ex},
lab/.style={shape=circle, font=\scriptsize}]
\foreach \i in {0,...,9}
{
\draw (\i*360/10: 3cm) node[circ] (\i) {$\scriptstyle \i$};
}
\foreach \x in {0,...,9} {
      \foreach \y in {\x,...,9} {
        \draw[thick] (\x) -- (\y);
        };
        };
\end{tikzpicture}
    $$
    The local view at each vertex $i$ of $\Gs_0$ is in a natural bijection with $\set A \times \set B \cong [3] \times [3]$. This provides a natural way to assign a dual product code $(\Ca \otimes \Cb)^\perp$ to each vertex, where both $\Ca$ and $\Cb$ are codes of length $3$.
    Demanding $k_\textnormal{A} + k_\textnormal{B} = \Delta = 3$, we must have $k_\textnormal{A} k_\textnormal{B} = 2$ to get a nontrivial code. So for each vertex in $\Gs_0$, two checks are added to $\code C_0$. This results in the Tanner graph of $\code C_0$ inheriting $4$-cycles from $(\Ca \otimes \Cb)^\perp$.

    As the graphs $\Gs_0$ and $\Gs_1$ have $45$ edges, the quantum Tanner code will have length $n=45$. If $\Ca$ and $\Cb$ are, respectively, the repetition code and the single parity-check code of length $3$, then the resulting code will have dimension $k\geq \frac{1}{2} 10 (3^2 - 4\cdot 1\cdot 2) = 5$. In practice, $k=7$, as two of the checks are redundant. The minimum distance of the code turns out to be $d = 4$. See Appendix \ref{ex:Tanner-graph-45-7-4} for the Tanner graph of the code. 
\end{example}

\begin{table}[t!]
 {\begin{center}
 \caption{Shortest 2-TNC codes given $\Delta$.$^a$
 }
 \label{tab:2-TNC-codes}
 \vspace{-2.0ex}
 \scalebox{1}{
     \begin{tabular}
     {c@{\hskip12pt}c@{\hskip12pt}c@{\hskip12pt}c@{\hskip12pt}c}
       \toprule 
       $\Delta$ & $[[n,k,d]]$ & $\hha$ & $\hhb$ & $w$\\
       \midrule
       $2$ & $[[10, 2, 3]]$ & $\left[\begin{smallmatrix}
           1 & 1
       \end{smallmatrix}\right]$ & $\left[\begin{smallmatrix}
           1 & 1
       \end{smallmatrix}\right]$ & $4$\\[0.3em]
       $3$ & $[[45, 7, 4]]$ & $\left[\begin{smallmatrix}
           1 & 0 & 1 \\ 0&1&1
       \end{smallmatrix}\right]$ & $\left[\begin{smallmatrix}
           1 & 1 & 1
       \end{smallmatrix}\right]$ & $6$\\[0.3em]
       $4$ & $[[136, 4, 4]]$ & $\left[\begin{smallmatrix}
           1&0&0&1 \\ 0&1&1&1
       \end{smallmatrix}\right]$ & $\left[\begin{smallmatrix}
           1&0&0&1 \\ 0&1&1&1
       \end{smallmatrix}\right]$ & $6.25$\\[0.3em]
       $4$ & $[[136, 38, 4]]$ & $\left[\begin{smallmatrix}
           1&1&1&1
       \end{smallmatrix}\right]$ & $\left[\begin{smallmatrix}
           1&0&0&1 \\ 0&1&0&1 \\ 0&0&1&1
       \end{smallmatrix}\right]$ & $8$\\[0.7em]
       $5$ & $[[325, 14, 9]]$ & $\left[\begin{smallmatrix}
           0&1&1&0&1 \\ 1&1&0&1&0
       \end{smallmatrix}\right]$ & $\left[\begin{smallmatrix}
           1&1&0&0&0 \\ 0&0&0&1&1 \\ 0&1&1&0&1
       \end{smallmatrix}\right]$ & $7$\\[0.7em]
       $6$ & $[[666, 4, \geq 13]]$ & $\left[\begin{smallmatrix}
            0& 1 &0 &1 &1 &0 \\
            1& 0& 0& 1& 0& 1\\
            1& 1& 1& 0& 0& 0
       \end{smallmatrix}\right]$ & $\left[\begin{smallmatrix}
           1& 1& 0& 0& 1& 0\\
           1& 0& 1& 1& 0& 0\\
           0& 1& 1& 0& 0& 1
       \end{smallmatrix}\right]$ & $9$ \\[0.7em]
       $6$ & $[[666, 8, \geq 13]]$ & $\left[\begin{smallmatrix}
           0 &1 &1 &0 &0 &1\\
           0 &0 &0 &1 &1 &1\\
           1 &1 &0 &1 &0 &0
       \end{smallmatrix}\right]$ & $\left[\begin{smallmatrix}
           0 &1 &1 &0 &0 &1\\
           1 &1 &0 &1 &0 &0\\
           0 &0 &0 &1 &1 &1
       \end{smallmatrix}\right]$ & $9$ \\[0.7em]
       $7$ & $[[1225, 32, \geq 13]]$ & $\left[\begin{smallmatrix}
           0 &1 &1 &0 &0 &1 &1\\
           0 &0 &1 &1 &1 &1 &0\\
           1 &0 &0 &1 &0 &1 &1
       \end{smallmatrix}\right]$ & $\left[\begin{smallmatrix}
           0 &1 &0 &0 &1 &1 &0\\
           0 &1 &0 &1 &0 &0 &1\\
           1 &0 &0 &0 &1 &0 &1\\
           0 &0 &1 &0 &0 &1 &1
       \end{smallmatrix}\right]$ & $12$ \\[1em]
       $7$ & $[[1225, 225, 6]]$ & $\left[\begin{smallmatrix}
           0 &1 &1 &1 &0 &1 &0\\
           1 &1 &0 &0 &1 &0 &1\\
       \end{smallmatrix}\right]$ & $\left[\begin{smallmatrix}
           0&1&0&1&0&0&0\\
           0&1&0&0&0&0&1\\
           1&0&1&0&0&0&0\\
           1&0&0&0&0&1&0\\
           1&1&0&0&1&0&0
       \end{smallmatrix}\right]$ & $8.8$ \\
       \bottomrule
     \end{tabular}}%
     
     \end{center}} 
 \footnotesize{$^a$Shortest quantum Tanner codes satisfying the 2-TNC condition, constructed as described in \cref{ex:short_2TNC} and exemplified in \cref{ex:45-7-4-2TNC}, apart from the $[[325, 14, 9]]$ code, which has $\set{A}=\{13, 2, 24, 17, 9\}$ and $\set{B} = \{0, 6, 20, 3, 23\}$. Average row weight is denoted by $w$. Different local codes, here given by $\hha$ and $\hhb$, can yield different values of $k$, $d$, and $w$.
 Parity-check matrices for the codes are available at \cite{nisqec26_1}.
 }
 
\vspace{-2ex}
\end{table}

The $[[45, 7, 4]]$ code described in \cref{ex:45-7-4-2TNC} is presented in \cref{tab:2-TNC-codes}, along with some other codes that are the shortest satisfying 2-TNC given their $\Delta$. Here, $\hha, \hhb$ are sampled from parity-check matrices of good codes with good dual codes. In particular, the shortened Hamming code is used for $\Delta=6$, and the Hamming code is used for the $[[1225, 32, \geq 13]]$ code. For the four shortest codes, the search was based on the dimension $k$ and the minimum distance $d$. The $[[325, 14, 9]]$ code was found by simulating the logical error rate of $10$ candidates and picking the best one, and the four longest codes were constructed by minimizing the maximum column weight of the matrix $\left[\begin{smallmatrix}
    \hhx \\ \hhz
\end{smallmatrix}\right]$.
Syndrome extraction circuits for some of these codes were constructed in \cite{RochaMostadRosnesLin2026_1sub}. All the codes in \cref{tab:2-TNC-codes} considered there were found to have such circuits of depth that meet the lower bound given by the maximum column and row weights of $\left[\begin{smallmatrix}
    \hhx \\ \hhz
\end{smallmatrix}\right]$.

The $[[325, 14, 9]]$ code provides an example of a code constructed from sets $\set A$ and $\set B$ different from those of \cref{ex:short_2TNC}, while the other codes of \cref{tab:2-TNC-codes} have $\set A$ and $\set B$ as described there. Since $\Gs_0$ and $\Gs_1$ are complete graphs for any code of length $\Delta^2+1$ satisfying 2-TNC, the expansion of $\Gs_0$ and $\Gs_1$ is the same for any $\set A$ and $\set B$. The same is not true, however, for the graphs $\Ga$ and $\Gb$, so one should expect some choices of $\set A$ and $\set B$ to be better than others.


\section{Proposed Decoder}

In this section, we outline the proposed iterative BP decoder. The main idea is to combine simple checks into more powerful checks  and then process these using a trellis-based MAP decoder as part of the check node processing of each
decoding iteration on the Tanner graph of the qLDPC code. A similar approach for improved decoding of binary linear codes was proposed in \cite{Rosnes2011}.

For some code families, such as quantum Tanner codes, HGP codes, and LP codes, there are inherent structures that give natural ways to combine checks for generalized decoding. See \cref{rem:ways-to-group} for a discussion of how this can be done for quantum Tanner codes. There might be less apparent structures in the Tanner graph that are worse for the decoder than the obvious ones, and some code families might not have an obvious way to combine checks, like GB codes, so we give a greedy algorithm in Algorithm \ref{alg:find-grouping-greedy} for combining checks according to strategy 4) of  \cref{rem:ways-to-group}.

%
Next, we describe MAP decoding of the generalized checks, and then give the pseudo-code of the proposed iterative BP decoder.

\begin{algorithm}[t!]
\label{alg:find-grouping-greedy}
    \caption{Check node combining}
    \KwData{An $m\times n$ parity-check matrix $\mat{H}_i$ and maximum size of blocks $r$.}
    \KwResult{A collection $\mathscr{H}$ of blocks of rows from $\mat H_i$ where each block contains at most $r$ elements, 
    and a Tanner graph $\G_i^{(r)}$. 
    }

     $\tilde{\mat H} \leftarrow \mat H_i$,
    $\mathscr{H} \leftarrow \emptyset$,
    and $a\leftarrow 0$

    $\Bcal \leftarrow \{r, \dots , r, r-1, \dots, r-1\}$ such that $\sum_{b\in \Bcal} b = m$\\
    $\tilde{m} \leftarrow m$

    \For{$b \in \Bcal$}{
        Set $a$ to a random number in $[\tilde{m}]$
        
        Set $\vect g$ to row number $a$ of $\tilde{\mat H}$
    
        Set $\hat{\mat H} \leftarrow \vect g$
    
        Delete row $a$ of $\tilde{\mat H}$,

        \For{$l \in [b-1]$}{
            {
            
            Compute $\tilde{\mat H} \trans{\vect g}$ and 
            find the index of the maximum value, arbitrarily breaking ties

            Remove the row $\vect r$ of that index from $\tilde{\mat H}$
        
            $\hat{\mat H} \leftarrow \left[\begin{smallmatrix}
                \hat{\mat H} \\ \vect r
            \end{smallmatrix}\right]$
        
            $\vect g \leftarrow \vect g {\lor} \vect r$
            }
        }
        Set $\tilde{m}$ to the number of rows of $\tilde{\mat H}$\\
        $\mathscr{H} \leftarrow \mathscr{H} \sqcup \hat{\mat H}$
        
    }
    
    




    



    



    

    Construct the Tanner graph $\G_i^{(r)}$ from $\mathscr{H}$, see \cref{sec:graph-prelim}.
    
    \textbf{return $\mathscr{H}, \G_i^{(r)}$} 
\end{algorithm}

\begin{algorithm}[t!]
\label{alg:SISO}
    \caption{Soft-input soft-output (\texttt{SISO})}
    \KwData{Local parity-check matrix $\mat{H}_c$, the corresponding syndrome vector $\bm s_c$,  the corresponding (ordered) messages from neighboring variables nodes $(\Gamma_0,\ldots,\Gamma_{n_c-1})$, and the corresponding coset leader $\bm u_c =\bm u_c(\bm s_c)$.}
        \KwResult{Set of outgoing (ordered) messages to neighboring variable nodes $(\Delta_0,\ldots,\Delta_{n_c-1})$}

Construct a trellis $\mathscr{T}_c = (\cup_{t=0}^{n_c} \mathcal{S}_t,\cup_{t=0}^{n_c-1} \mathcal{E}_t)$, where $e \in \mathcal{E}_{t-1}$, with label $\ell(e)$, connects a state (or vertex) $s_\textnormal{l}(e) \in \mathcal{S}_{t-1}$ to a state $s_\textnormal{r}(e) \in \mathcal{S}_{t}$,   based on $\mat{H}_c$. \\
\For{$t \in [1:n_c]$}{
\For{$s \in \mathcal{S}_t$}{
$\alpha_t(s) 
\leftarrow \maxstar^*_{\substack{e \in \mathcal{E}^{(t-1)}: s_\textnormal{r}(e) = s}}$ $\quad \quad \quad \left( \alpha_{t-1}(s_\textnormal{l}(e)) +  (\ell(e)+u_{c, t-1}) \Gamma_{t-1}\right)$}} 
%
\For{$t \in [n_c:1]$}{
\For{$s \in \mathcal{S}_t$}{
$\beta_{t-1}(s) 
\leftarrow \maxstar^*_{\substack{e \in \mathcal{E}^{(t-1)}: s_\textnormal{l}(e) = s}}$ $\left( \beta_{t}(s_\textnormal{r}(e)) + ( \ell(e)+u_{c, t-1} )\Gamma_{t-1} \right)$}} 
\For{$t \in [1:n_c]$}{
$\Delta_{t-1} \leftarrow \maxstar^*_{\substack{e \in \mathcal{E}^{(t-1)}:\\ \ell(e) + u_{c, t-1} = 0}} \left( \alpha_{t-1}(s_\textnormal{l}(e))  + \beta_{t}(s_\textnormal{r}(e)) \right)
  - \maxstar^*_{\substack{e \in \mathcal{E}^{(t-1)}:\\ \ell(e) + u_{c, t-1} = 1}} \left( \alpha_{t-1}(s_\textnormal{l}(e)) + \beta_{t}(s_\textnormal{r}(e)) \right)$} 
   \textbf{return $(\Delta_0,\ldots,\Delta_{n_c-1})$} 
\end{algorithm}

\begin{algorithm}[t!]
\label{alg:BP4_mem}
    \caption{Generalized BP$_{4}$ with memory (\texttt{GMBP$_4$})}
    \KwData{A syndrome $\vect s$,
    a PCM $\mat H 
    = \left[\begin{smallmatrix}
        \mat H_0 \\ \mat H_1
    \end{smallmatrix}\right]$ where ${\mat H_0}_{ij} \in \{\mat I, \mat X\}$ and ${\mat H_1}_{ij} \in \{\mat I, \mat Z\}$, 
    graphs $\G_i^{(r)} = (\vv_i^{(r)}= {\vv^{(r)}_{\textnormal{c}, i}}\sqcup \vv_\textnormal{v}, \ee_i^{(r)}), \ i \in \{0, 1\}$,
    local (binary) PCMs $\Hscr = \left\{\mat H_c\right\}_{c \in \vv^{(r)}_{\textnormal{c}}}$ with coset leaders $\vect{u}_c = \vect{u}_c(\vect s_c)$, $c \in \vv^{(r)}$, 
    a maximum number of iterations $T_\textnormal{max}$,
     initial LLRs $\{\Lambda_v^\Wsf\in \Rbb\}_{v\in \vv_{\vtn}, \mat{W}\in\{\mat X, \mat Y, \mat Z\}}$,
    and
    a \textit{flag} $\in \{0, 1\}$.
    }
    \KwResult{%
    A Pauli string $\hat{\bm{\mat E}}$ such that $\langle \hat{\bm{\mat E}} , \mat H_{i,:}\rangle = s_i$ for all $i$ (success),
    or a Pauli string $\hat{\bm{\mat E}}$ such that $\langle \hat{\bm{\mat E}} , \mat H_{i,:}\rangle \neq s_i$ for some $i$ (failure).
    }

    Initialization: 
    Construct the Tanner graph 
    \begin{align*}
    \G^{(r)} 
    &= 
    \left(\vv^{(r)}= {\vv^{(r)}_\textnormal{c}} \sqcup {\vv_\textnormal{v}}, \ee^{(r)}\right) 
    \end{align*}
    where ${\vv^{(r)}_\textnormal{c}} = {\vv^{(r)}_{\textnormal{c},0}} \sqcup {\vv^{(r)}_{\textnormal{c},1}}$ and $\ee^{(r)} = \ee^{(r)}_0 \sqcup \ee^{(r)}_1$

    $\Gamma^{\mat W}_{v \to c} \leftarrow \Lambda^{\mat W}_v$, $v\in \vv_\textnormal{v} \cong [n]$, $c \in \ee^{(r)}(v)$, $\mat W \in \{\mat X, \mat Y, \mat Z\}$ 
    %
    
    \For{$t \in [T_{\max}]$}{


        \For{$i \in \{0,1\}, c \in {\vv_{\textnormal{c},i}^{(r)}}, v \in \ee^{(r)}(c)$}{
        $\Gamma_{v\to c} \leftarrow
            \log (1 + {\mathrm e}^{-\Gamma_{v\to c}^{\Wsf_i}})+
            \min(\Gamma_{v\to c}^\Ysf, \Gamma_{v \to c}^{\Wsf_{\Bar{\imath}}}) 
            - \log \left(1 + {\mathrm e}^{-\left|\Gamma_{v \to c}^\Ysf - \Gamma_{v \to            c}^{\Wsf_{\Bar{\imath}}}\right|}\right)$

        }
        
        \For{$c \in \vv_\ctn^{(r)}$}{
            $\{ \Delta_{c \to v} \}_{v \in \ee^{(r)}(c)} \leftarrow \texttt{SISO}\left(\mat H_c, \vect s_c, \{\Gamma_{v \to c}\}_{v\in \ee^{(r)}(c)}, \vect{u}_c(\vect{s}_c)\right)$
        }
    
        \For{$v\in \vv_\vtn, \mat W \in \{\mat X, \mat Y, \mat Z\}$}{
            $\Gamma_v^\mat{W} \leftarrow \Lambda^{\mat W}_v + \frac{1}{\alpha} \sum_{\substack{c \in \ee^{(r)}(v)\\\langle\mat W, \mat H_{cv}\rangle = 1}} \Delta_{c \to v}$
        }
    
        (Hard decision:) Let $\hat{\bm{\mat E}} \in \{\Isf, \Xsf, \Ysf, \Zsf\}^{n}$ be the vector where $\hat{{\mat E}}_v = \Isf$ if $\Gamma_v^\Wsf > 0$ for all $\Wsf\in \{\Xsf, \Ysf, \Zsf\}$ and $\hat{{\mat E}}_v = \argmin_{\Wsf\in \{\Xsf, \Ysf, \Zsf\}} \Gamma_v^\Wsf$ otherwise.

        \If{$\langle \hat{\bm{\mat E}} , \mat H_{i,:}\rangle = s_i$ \textnormal{for all rows} $\mat H_{i,:} \subseteq \mat H$}{\textbf{return} $\hat{\bm{\mat E}}$}

        \For{$v \in \vv_\vtn, c\in \ee^{(r)}(v), \Wsf\in \{\Xsf, \Ysf, \Zsf\}$}{
            $\Gamma_{v\to c}^\mat{W} \leftarrow \Gamma^{\mat W}_{v} - \frac{1}{\alpha}\langle\mat W, \mat H_{cv}\rangle \Delta_{c \to v}$ \label{line:17}
        }
    }

    \If{\textit{flag}}{
        \textbf{return} \texttt{OSD-$1$}$(\mat H, \{\Gamma_v^\Wsf\}_{v\in \ee^{(r)}, \Wsf\in \{\Xsf, \Ysf, \Zsf\}})$
    }

    \textbf{return} $\hat{\bm{\mat E}}$ (failure) 

\end{algorithm}


%

\subsection{Trellis-Based MAP decoder}
\label{sec:trellis-based-MAP}

The proposed trellis-based MAP decoder is outlined in Algorithm~\ref{alg:SISO}. It follows the standard forward/backward recursions of the BCJR algorithm in the logarithm domain, but  adapted for a nonzero syndrome vector, i.e., the code trellis is modified by adding a coset leader (from the code coset corresponding to the syndrome vector) to all labelled paths (or codewords) in the trellis.  
A (minimal) code trellis for a linear block code can for instance be build using partial syndromes as outlined in \cite{Wolf1974}. We note that the complexity can be reduced by adopting a dual code approach as detailed for classical  codes in \cite{Montorsi2001} when the generalized check codes have rates above $\nicefrac{1}{2}$.

The complexity of Algorithm~\ref{alg:SISO} is proportional to the number of edges in the underlying trellis. The maximum number of trellis states for each depth is $2^{\min(k_c, n_c-k_c)}$ which can be reached at depth $\min(k_c, n_c-k_c)$ \cite{Wolf1974}, if the input $\mat H_c$ is an $(n_c-k_c)\times n_c$ full-rank matrix. Moreover,  there are $k_c$ depths with two outgoing edges from each state, and $n_c-k_c$ depths with a single outgoing edge from each state. In general, the total number of edges can be seen to be at most 
\begin{align*}
\begin{cases}
2(2^{k_c+1}-2) + 2^{k_c} (n_c-2k_c) & \\
 = 2^{k_c} (4+n_c-2k_c) -4 & \hspace{-1cm} \text{if $ k_{c} \leq n_c-k_c$}\\
2(2^{n_c-k_c+1}-2) + 2 \cdot 2^{n_c-k_c} (n_c-2(n_c-k_c)) & \\
= 2^{n_c-k_c+1} (2-n_c+2k_c) -4 & \hspace{-1cm} \text{otherwise}. \end{cases}
\end{align*}
When the checks of a quantum Tanner code are fully combined based on the local code structure, $n_c = \Delta^2$ and $k_c=k_\textnormal{A} k_\textnormal{B}$, while a partial combination gives $n_c \leq \Delta^2$ and $k_c\leq k_\textnormal{A} k_\textnormal{B}$. Furthermore, if the checks coming from $\hha$ in $\hha\otimes \hhb$ are combined, then the different $n_c$ will be the products of $\Delta$ by the row weights of $\hhb$. In \cref{tab:codes}, we tabulate   the total number of edges from the formula above, averaged over the generalized checks, for the different codes studied in this work.

\subsection{Improved BP Decoding Using Generalized Checks}
The pseudo-code of the proposed decoder is given in Algorithm~\ref{alg:BP4_mem}. It is based on the BP$_4$ decoder with memory effects (MBP$_4$) \cite[Alg.~1]{KuoLai22_2}, and has a flag for OSD post-processing. The  selection of the scaling parameter $\alpha$ is important to optimize the performance for either the relatively high or the low error rate region, with an $\alpha > 1.0$ typically giving improved performance for lower error rates (in the so-called error floor region). 
Note that we keep the scaling parameter $\alpha$ in \cref{line:17}, rather than removing it, as they do in \cite[Alg.~1]{KuoLai22_2}, where this is referred to as fixed inhibition. This is because we observe better performance that way for the codes we have simulated.

The proposed decoder has an inherent complexity/performance trade-off in the sense that the size of the generalized checks from either Algorithm~\ref{alg:find-grouping-greedy} or from the local codes can be made smaller (as discussed in \cref{sec:cycles-QT}) in order to reduce the complexity of running Algorithm~\ref{alg:SISO}, which may lead to poorer decoding performance. This trade-off will be explored in the numerical results. Moreover, in order to lower the complexity further, we can run  a hybrid version where the standard  MBP$_4$ 
decoder is run first without OSD post-processing, and if it does not converge, we run Algorithm~\ref{alg:BP4_mem} with or without OSD post-processing.

The asymptotic complexity of our proposed decoder for quantum Tanner codes, Algorithm \ref{alg:BP4_mem}, is the same as for the standard (M)BP$_4$ decoder, i.e., $O(n)$, when the size of the local codes and the maximum number of iterations $T_{\textnormal{max}}$ is kept constant. Constant-size local codes are indeed realistic, as the quantum Tanner codes are asymptotically good qLDPC codes in this case, at least if the constant $\Delta$ is large enough \cite{LeverrierZemor22_1}. Note that all the inner for-loops are parallelizable in Algorithm \ref{alg:BP4_mem} (and also in Algorithm \ref{alg:SISO}). 
However, the constant cost of the trellis-based MAP decoder of Algorithm \ref{alg:SISO} can be daunting if the local codes are large, see \cref{sec:trellis-based-MAP} and \cref{tab:codes} (last column).
The OSD-$1$ we use here is the standard one, see~\cite{PanteleevKalachev21_1}. The complexity of OSD-$\omega$ is $O(n^3 + n2^\omega)$.


\section{Numerical Results}\label{sec:numerical-results}

\begin{table*}[t!]
 {\begin{center}
 \caption{List of Simulated Codes$^b$}
 \label{tab:codes}
 \vspace{-2.0ex}
 \scalebox{0.85}{
     \begin{tabular}
     {ccccccc}
       \toprule 
       \multirow{2}{1.1cm}{$[[n,k,d]]$} & \multirow{2}{1.4cm}{Construction} & \multicolumn{2}{c}{Row weight} & \multirow{2}{0.8cm}{Figure} & \multirow{2}{0.8cm}{Source} & Average trellis complexity \\
       &&average & range &&& upper bound ($r$ combined checks)\\
       \midrule
       $[[144, 12, 12]]$ & BB &$6$&$6$&\ref{fig:logical_error_rate_144}&\cite{BravyiCrossGambettaMaslovRallYoder2024} & $12(1)$, $14503(8)$\\[0.3em]
       $[[144, 12, 12]]$ & GB &$8$&$8$&\ref{fig:logical_error_rate_144}&\cite{MostadRosnesLin2025_1} & $16(1)$, $608(4)$, $15356(8)$\\[0.3em]
       $[[144, 12, 11]]$ & QT &$9$&$9$&\ref{fig:logical_error_rate_144}&\cite{LeverrierRozendaalZemor25_1sub} & $18(1)$, $22652(9)$\\[0.3em]
       $[[324, 8, \leq26]]$ & QT &$10.5$&$7$--$16$&\ref{fig:2-TNC-plot}&\cite{Wang-etal26_1sub} & $21.1(1)$, $11260(9)$\\[0.3em]
       $[[325, 14, 9]]$ & QT &$7$&$6$-$9$&\ref{fig:2-TNC-plot}& \cref{tab:2-TNC-codes} & $14(1)$, $1084(6)$\\[0.3em]
       $[[377, 25, 5]]$ & HGP & $6.75$ & $6$--$7$ & \ref{fig:logical_error_rate_432} & \cite{Liu2026database} & $13.5(1)$, $14402(8)$\\[0.3em]
       $[[416, 18, \leq 22]]$ & LP & $8$ & $8$ & \ref{fig:logical_error_rate_432} & \cite{Liu2026database} & $16(1)$, $1870(5)$\\[0.3em]
       $[[432, 16, \leq 26]]$ & QT & $13.31$ & $12$--$16$ & \ref{fig:logical_error_rate_432}, \ref{fig:logical_error_rate_432_different_level} & \cite{Wang-etal26_1sub} & $26.6(1)$, $449(3)$, $659(4)$, $3448(6)$, $212988(12)$\\[0.3em]
       $[[432, 20, \leq 22]]$ & QT & $9$ & $9$ & \ref{fig:logical_error_rate_432} & \cite{LeverrierRozendaalZemor25_1sub} & $18(1)$, $20476(9)$\\[0.3em]
       $[[686,34,13]]$ &QT&$13$&$12$--$16$ & \ref{fig:logical_error_rate_686} & \cite{MostadRosnesLin25_1} & $26(1)$, $221180(12)$\\[0.3em]
       $[[686,28,14]]$ &QT&$13$&$12$--$16$ & \ref{fig:logical_error_rate_686} & \cite{MostadRosnesLin25_1}& $26(1)$, $221180(12)$\\[0.3em]
       $[[688,30, \geq 19]]$ & GB & $16$ & $16$ & \ref{fig:logical_error_rate_686} & \cite{MostadRosnesLin25_1} & $32(1)$, $52243(8)$\\
       \bottomrule
     \end{tabular}}%
     
     \end{center}} 
 \footnotesize{$^b$Here, QT stands for quantum Tanner code. For $r=1$, we show the bound for the complexity of decoding the on dual code locally, as this is less complex and equivalent to the box-plus combining rule (see \cref{line:boxplus} in Algorithm~\ref{alg:Relay-BP4}). 
 For the $[[432, 16, \leq 26]]$ code, we give the average trellis complexity of the best performing grouping for $r=3$ and $4$ in \cref{fig:logical_error_rate_432_different_level}. Note that when $r$ is larger there are fewer checks to process in each BP decoding iteration.
 }
\vspace{-2ex}
\end{table*} 

Next, we compare the logical error rate decoding performance of qLDPC codes on the depolarizing channel, under MBP$_4$ with and without OSD of order $1$   post-processing as in~\cite{PanteleevKalachev21_1}; with and without MAP decoding of generalized check nodes.\footnote{The BP$_4$ decoding part of our simulations is based on the implementation at \url{https://github.com/kit-cel/Quantum-Neural-BP4-demo} with the memory and OSD part implemented in house, together with the MAP decoder of the generalized checks.}
We select $\alpha=1.6$ ($\alpha=1.5$ for the BB code) for our simulations, unless specified otherwise, and present results for three lengths $n=144$, $432/416/377$, and $686/688$.  Unless specified otherwise, we use a maximum of $6$ BP iterations in the simulations. For generalized decoding, we use a hybrid version where we first run the conventional MBP$_4$ decoder and if it does not converge, we run Algorithm~\ref{alg:BP4_mem}, both using a maximum of $6$ iterations. When using OSD post-processing, this step is only run after the second step of generalized decoding.  We also compare with the recently introduced  Relay-BP decoding algorithm in \cite[Alg.~1]{Mueller-etal25_1sub}. Here, we adapt Relay-BP to a quaternary version referred to as  Relay-BP$_4$ and given as Algorithm~\ref{alg:Relay-BP4} in Appendix~\ref{app:relay_bp}.

\begin{figure}[t]
\begin{center}
\subfloat{\includegraphics[width=0.95\columnwidth]{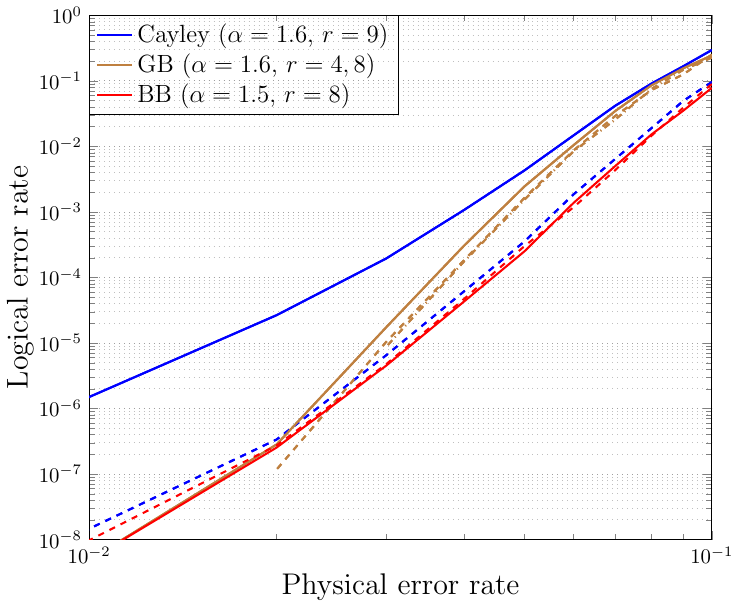}} 
\end{center}
\vspace{-3ex}
\caption{Comparing the logical error rate performance of (generalized) MBP$_4$+OSD-$1$ decoding (with $\alpha=1.6$; $\alpha=1.5$ for the BB code) of three different $[[144,12]]$ qLDPC codes on the depolarizing channel. Solid curves are for standard MBP$_4$+OSD-$1$ decoding, while dashed and dash-dotted curves are for generalized  MBP$_4$+OSD-$1$ decoding.
} \label{fig:logical_error_rate_144}
\vspace{-2ex}
\end{figure}


\begin{figure}[t]
\begin{center}
\subfloat{\includegraphics[width=0.95\columnwidth]{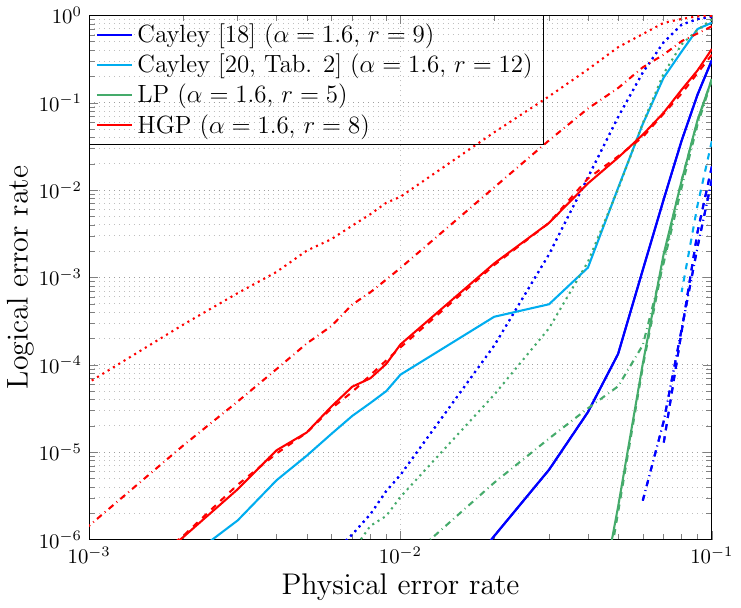}} 
\end{center}
\vspace{-3ex}
\caption{Comparing the logical error rate performance of (generalized) MBP$_4$ decoding with and without OSD-$1$ post-processing (with $\alpha=1.6$) of the $[[432,20,\leq 22]]$ quantum Tanner code from \cite{LeverrierRozendaalZemor25_1sub} (dark blue curves), the $[[432,16,\leq 26]]$ quantum Tanner code from \cite[Tab.~2]{Wang-etal26_1sub} (light blue curves), and the $[[416,18,\leq 22]]$ LP code and the $[[377,25,5]]$ HGP code from the GitHub qLDPC code repository~\cite{Liu2026database} (green and red curves) on the depolarizing channel. Solid curves are for standard MBP$_4$+OSD-$1$ decoding, dashed curves are for  generalized  MBP$_4$+OSD-$1$ decoding, while curves with dots are for decoding without OSD-$1$ post-processing (dotted curves for standard MBP$_4$ decoding with a maximum of $6$ BP iterations and dash-dotted curves for generalized  MBP$_4$ decoding with a maximum of $25$ BP iterations).
} \label{fig:logical_error_rate_432}
\vspace{-2ex}
\end{figure}

In Fig.~\ref{fig:logical_error_rate_144}, we compare the performance of Algorithm~\ref{alg:BP4_mem} (MBP$_4$+OSD-$1$), without (solid curves) and with (dashed curves) MAP decoding of local codes. We compare the performance of the $[144,12,12]$ BB code  from \cite{BravyiCrossGambettaMaslovRallYoder2024} (red curves) with that of an optimized $[144,12,12]$ GB code from \cite{MostadRosnesLin2025_1} (brown curves) and the $[[144,12,11]]$ quantum Tanner code from \cite{LeverrierRozendaalZemor25_1sub} (blue curves). 
For the quantum Tanner code, we consider the full grouping (see \cref{rem:ways-to-group}(1)), while Algorithm \ref{alg:find-grouping-greedy} is used to find a grouping for the GB and BB code (see also \cref{rem:ways-to-group}(4)).
From the figure, we can observe a significant performance improvement for the quantum Tanner code, while only modest gain for the GB and BB codes (the number of combined checks $r$ is given in the legend). Interestingly, the quantum Tanner code performs comparably to the BB code under generalized (hybrid) decoding. 
 
 In \cref{fig:logical_error_rate_432}, we show the corresponding results for length $n=432/416/377$. We compare the  performance of (generalized) MBP$_4$ decoding with and without OSD of order $1$ post-processing  (with $\alpha=1.6$) of the $[[432,20,\leq 22]]$ quantum Tanner code from \cite{LeverrierRozendaalZemor25_1sub} (dark blue curves), the $[[432,16,\leq 26]]$ quantum Tanner code from \cite[Tab.~2]{Wang-etal26_1sub} (light blue curves), and the $[[416,18,\leq 22]]$ LP code and the $[[377,25,5]]$ HGP  code from the GitHub qLDPC code repository~\cite{Liu2026database} (green and red curves). Solid curves are for standard MBP$_4$+OSD-$1$ decoding, dashed curves are for  generalized  MBP$_4$+OSD-$1$ decoding, while curves with dots are for decoding without OSD of order $1$ post-processing  (dotted curves for standard MBP$_4$ decoding with a maximum of $6$ BP iterations and dash-dotted curves for generalized  MBP$_4$ decoding with a maximum of $25$ BP iterations). 
The full grouping is used for the quantum Tanner codes (see \cref{rem:ways-to-group}(1)), while Algorithm \ref{alg:find-grouping-greedy} is used globally for the LP code and locally for the HGP code. 
 As for length $n=144$, for both quantum Tanner codes, we can observe a significant performance improvement with generalized decoding, while  for the LP and HGP codes there is no gain compared to conventional MBP$_4$+OSD-$1$ decoding. Moreover, it can be observed that OSD of order $1$ post-processing can boost the performance  for both conventional and generalized MBP$_4$ decoding.

\begin{figure}[t]
\begin{center}
\subfloat{\includegraphics[width=0.95\columnwidth]{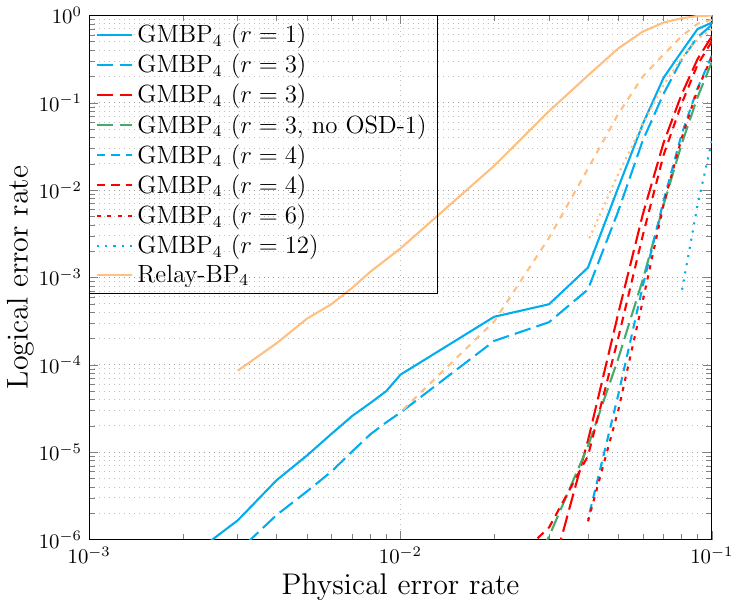}} 
\end{center}
\vspace{-3ex}
\caption{Comparing the logical error rate performance of (generalized) MBP$_4$+OSD-$1$ decoding (with $\alpha=1.6$) of the $[[432,16,\leq 26]]$ quantum Tanner code from \cite[Tab.~2]{Wang-etal26_1sub} 
for different values of $r$ on the depolarizing channel. 
The solid light blue curve is
for standard MBP$_4$+OSD-$1$ decoding, and dashed to dotted light blue and red curves correspond (in this order) to increasing values of $r$  ($r=3,4,6,12$). As a comparison, the performance of Relay-BP$_4$ (Algorithm~\ref{alg:Relay-BP4}) is shown  with different values for the maximum number of relay legs ($R$ in Algorithm~\ref{alg:Relay-BP4}) and the maximum number of BP iterations for each leg  ($T_r$ in  Algorithm~\ref{alg:Relay-BP4}) (orange curves; the orange dotted curve has the same $R$ and $T_r$ as the dashed one, but with OSD-$1$ post-processing). The green curve is with no OSD-$1$ post-processing, but with a maximum of $50$ BP iterations. 
} \label{fig:logical_error_rate_432_different_level}
\vspace{-2ex}
\end{figure}

Next, in \cref{fig:logical_error_rate_432_different_level}, we show the  performance of (generalized) MBP$_4$+OSD-$1$ decoding (with $\alpha=1.6$) of the $[[432,16,\leq 26]]$ quantum Tanner code from \cref{fig:logical_error_rate_432}, for different values of $r$ to assess the performance-complexity trade-off of the proposed generalized MBP$_4$ decoder. The solid light  blue curve is for standard MBP$_4$+OSD-$1$ decoding, and dashed to dotted light blue and red curves correspond   (in this order) to increasing values of $r$  ($r=3,4,6,12$). The light blue curves are for structured partial and full groupings, while the red curves are for unstructured local ones made with Algorithm \ref{alg:find-grouping-greedy}, see \cref{rem:ways-to-group}.
Interestingly, there is a significant performance improvement even for unstructured $r=3$ (while there is only limited improvement for the other $r=3$). Increasing $r$ to $4$ improves the unstructured grouping only  marginally, while the structured one is better. Increasing $r$ to $6$ does not bring any further gain over the best $r=4$ curve, while increasing it all the way to $12$ (structured full grouping) improves performance even further, but at the cost of a significant increase in decoding complexity (see \cref{tab:codes}).
This is likely because the $4\times 7$ parity-check matrices used have more $4$-cycles than the $3\times 7$ ones, but in such a way that in the Kronecker product of these two, many of the $4$-cycles are within blocks of three rows. This also shows that Algorithm \ref{alg:find-grouping-greedy} can be both better and worse than a structured grouping.
As a comparison, the performance of Relay-BP$_4$ (Algorithm~\ref{alg:Relay-BP4}) 
with and without OSD of order $1$ post-processing is shown as well (orange curves) with different values for the maximum number of relay legs ($R$ in Algorithm~\ref{alg:Relay-BP4}) and the maximum number of BP iterations for each leg  ($T_r$ in  Algorithm~\ref{alg:Relay-BP4}). The memory strengths for each leg are drawn independently and uniformly from a common range $[\gamma_\textnormal{c} -  \gamma_\textnormal{w}/2, \gamma_\textnormal{c} + \gamma_\textnormal{w}/2]$ with $\gamma_\textnormal{c} = 0.3$ (sweep optimized) and fixed $\gamma_\textnormal{w}=0.66$, and the number of solutions sought ($S$ in  Algorithm~\ref{alg:Relay-BP4}) is set equal to $R$. With modest values of $R=5$ and $T_r=6$, and no OSD post-processing, the performance (solid orange curve) is much worse than for MBP$_4$+OSD-$1$  decoding (light blue solid curve). Increasing $R$ to $25$ and $T_r$ to $30$ (which entails a significant increase in the overall number of BP iterations), yields a decent performance improvement (dashed orange curve), but it is still worse than the light blue solid curve. Adding OSD of order $1$ post-processing on top, i.e., after each relay leg, does bring an additional benefit (dotted orange curve), but at a high computational cost.
We also show a dashed green curve for $r=3$ (with the same unstructured local groupings as for the red dashed curve), but with a maximum of $50$ BP iterations and with no OSD-$1$ post-processing. Notably, this curve is better than the red dashed curve for error rates down to $10^{-5}$ where there is a crossing, i.e., removing the OSD post-processing and increasing the maximum number of BP iterations from $6$ to $50$ improves performance. Moreover,  the performance is significantly better compared to the Relay-BP$_4$ curve with $R=25$ and $T_r = 30$ (dashed orange curve), and with what appears to also be a lower decoding cost as the average  trellis complexity of the generalized checks is $449$ ($r=3$; see \cref{tab:codes}) which should be compared to $26.6 \cdot 3 \approx 80$ ($r=1$, but with $3$ times as many checks as for $r=3$), while the maximum number of BP iterations is about $7.5$ times higher for   the Relay-BP$_4$ curve (orange dashed curve). Finally, we remark that  adding OSD-$1$ post-processing on top of the green dashed curve (not shown) improves the performance only slightly, which is in contrast to standard MBP$_4$ decoding where typically it is difficult to come close to MBP$_4$+OSD-$1$ with a small number of BP iterations,  by removing the OSD post-processing stage and increasing the number of BP iterations.

\begin{figure}[t]
\begin{center}
\subfloat{\includegraphics[width=0.95\columnwidth]{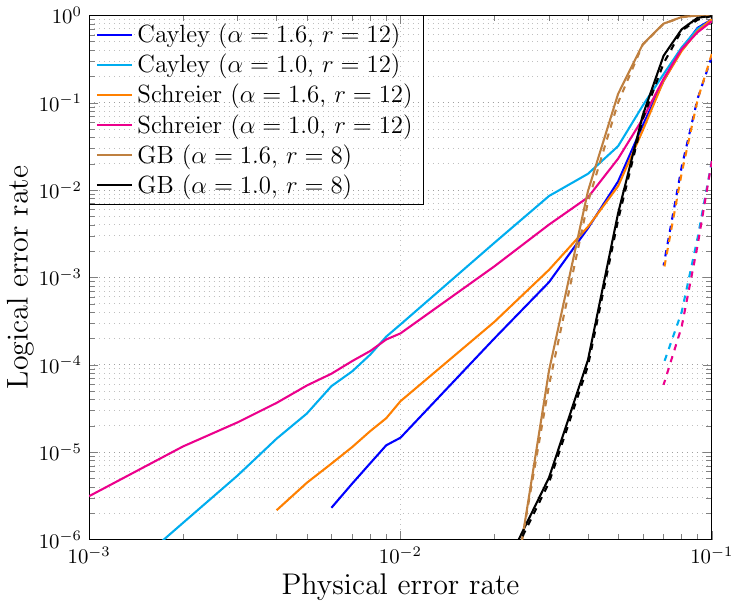}} 
\end{center}
\vspace{-3ex}
\caption{Comparing the logical error rate performance  of (generalized) MBP$_4$+OSD-$1$ decoding of two (generalized) quantum Tanner codes based on Cayley  (dark and light blue curves) and Schreier (orange and magenta curves) graphs, respectively, and a GB code (brown and black curves) on the depolarizing channel for $n=686$ ($688$ for the GB code). Solid curves are for standard MBP$_4$+OSD-$1$ decoding, while dashed curves are for  generalized  MBP$_4$+OSD-$1$ decoding. For all three codes, two different values for $\alpha$ ($1.0$ and $1.6$) are compared.
} \label{fig:logical_error_rate_686}
\vspace{-2ex}
\end{figure}


In Fig.~\ref{fig:logical_error_rate_686}, we show the corresponding results for length $n=686/688$. Here, we compare the performance of an optimized  $[[686,28,14]]$ generalized quantum Tanner code (from the Schreier construction; orange and magenta curves) with that of an optimized  $[[686,34,13]]$ quantum Tanner code (from the Cayley construction; dark and light blue curves), both taken from~\cite[Tab.~I]{MostadRosnesLin25_1}. For comparison, we also show the performance of an optimized $[[688,30,\geq 19]]$ GB code  \cite[Tab.~II]{MostadRosnesLin25_1} (brown and black curves).  The performance with two different values of $\alpha$ is shown ($\alpha=1.0$ and $1.6$). 
As before, we use the full grouping for the quantum Tanner codes and Algorithm \ref{alg:find-grouping-greedy} to find an unstructured grouping for the GB code.
Again, generalized decoding (dashed curves) gives enhanced  performance (solid curves are for conventional decoding) for the two Tanner codes, while only a modest gain is observed for the GB code. Interestingly, here the Tanner codes outperform the GB code by a large margin. We also observe the effect of the two different values of $\alpha$, with the smaller one showing improved performance for higher error rates. This  is consistent with previous results from the literature.

Finally, in \cref{fig:2-TNC-plot}, we show corresponding results for the $[[325, 14, 9]]$ quantum Tanner code (green curves) of \cref{tab:2-TNC-codes}, which has the shortest length $n$ possible for quantum Tanner codes satisfying 2-TNC with $\Delta=5$. We compare with the $[[324, 8, \leq26]]$ quantum Tanner code (magenta curves) from \cite[Tab.~2]{Wang-etal26_1sub}, a code we believe has a much larger minimum distance, at the cost of lower rate and higher row weight. As expected, the $[[324, 8, \leq26]]$ code performs better than the $[[325, 14, 9]]$ code when grouping $r=9$ and $r=6$ checks at a time, respectively, according to strategy 1) of \cref{rem:ways-to-group}. However, for the conventional decoder (solid curves), we see the opposite behavior. We attribute this to the 2-TNC code being easier for a BP decoder to decode well, as it has far fewer $4$-cycles than the other code. This also aligns with the observation that the generalized decoder shows only a limited gain in performance for the green curve.

\begin{figure}
    \centering
    \includegraphics[width=0.95\linewidth]{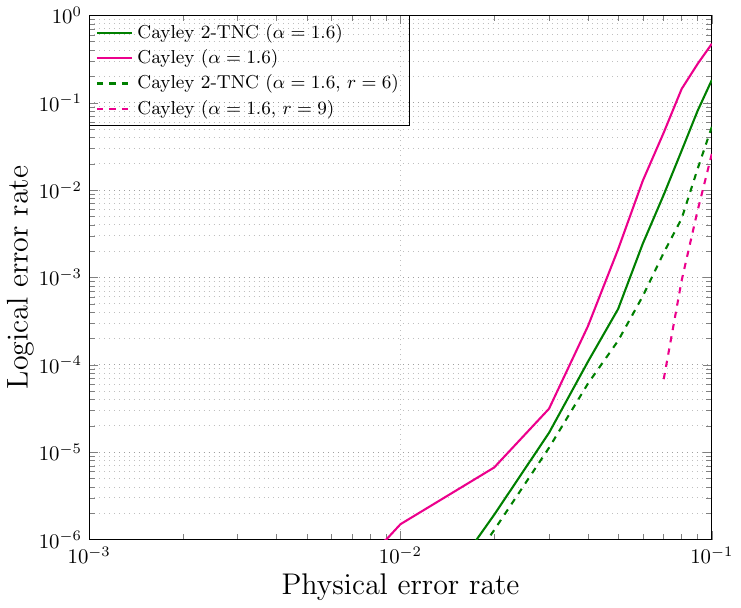}
    \caption{Comparing the logical error rate performance of (generalized) MBP$_4$+OSD-$1$ decoding of two quantum Tanner codes. Solid and dashed curves are for standard and generalized MBP$_4$+OSD-$1$ decoding, respectively. The green curves are for a $[[325, 14, 9]]$ code that satisfies 2-TNC, while the magenta curves are for a $[[324, 8, \leq 26]]$ code that does not satisfy 2-TNC.
    }
    \label{fig:2-TNC-plot}
\end{figure}

\begin{algorithm}[t!]\label{alg:Relay-BP4}
   \caption{Relay-BP$_4$}
   \KwData{A syndrome $\vect s$,
    a PCM $\mat H 
    = \left[\begin{smallmatrix}
        \mat H_0 \\ \mat H_1
    \end{smallmatrix}\right]$ where ${\mat H_0}_{ij} \in \{\mat I, \mat X\}$ and ${\mat H_1}_{ij} \in \{\mat I, \mat Z\}$,
    initial LLRs $\{\Lambda_v^\Wsf\in \Rbb\}_{v\in \vv_{\vtn}, \mat{W}\in\{\mat X, \mat Y, \mat Z\}}$, a number of solutions to be
found $S$, 
    a maximum number of legs for the relay $R$,
    a maximum number of iterations per leg $T_r$,
    and
    memory strengths for each leg $\{\gamma(r)\}_{r\in [R]}$.
    }
   \KwResult{%
    A Pauli string $\hat{\bm{\mat E}}$ such that $\langle \hat{\bm{\mat E}} , \mat H_{i,:}\rangle = s_i$ for all $i$ (success),
    or a Pauli string $\hat{\bm{\mat E}}$ such that $\langle \hat{\bm{\mat E}} , \mat H_{i,:}\rangle \neq s_i$ for some $i$ (failure).}

    $\hat{\bm{\mat E}} \leftarrow \emptyset$, $s \leftarrow 0$, $\omega_{\hat{\bm{\mat E}}} \leftarrow 0$

    $\Gamma^{\mat W}_{v} \leftarrow \Lambda^{\mat W}_v$, $v\in \vv_\textnormal{v} \cong [n]$, $\mat W \in \{\mat X, \mat Y, \mat Z\}$

    \For{$r \in [R]$}{
    \For{$v\in [n], \mat W\in \{\Xsf, \Ysf, \Zsf\}$}{
    $\Gamma_{v\to c}^\Wsf \leftarrow \Lambda_v^\Wsf$ 
    }

        \For{$i \in \{0,1\}, c \in {\vv_{\textnormal{c},i}}, v \in \ee(c)$}{
        $\Gamma_{v\to c} \leftarrow
            \log (1 + {\mathrm e}^{-\Gamma_{v\to c}^{\Wsf_i}})+
            \min(\Gamma_{v\to c}^\Ysf, \Gamma_{v \to c}^{\Wsf_{\Bar{\imath}}}) 
            - \log \left(1 + {\mathrm e}^{-\left|\Gamma_{v \to c}^\Ysf - \Gamma_{v \to            c}^{\Wsf_{\Bar{\imath}}}\right|}\right)$
        }

    \For{$t\in [T_r]$}{
    \For{$v\in [n], \Wsf \in \{\Xsf, \Ysf, \Zsf\}$}{
    $\Lambda_v^\Wsf(t) \leftarrow (1-\gamma_v(r)) \Lambda_v^\Wsf + \gamma_v(r) \Gamma_v^\Wsf$
    }


    \For{$c\in \vv_\textnormal{c}, v \in \ee(c)$}{
    $\Delta_{c\to v} \leftarrow 
    (-1)^{s_c} \boxplus_{v'\in \ee(c) \setminus \{v\}}
    \Gamma_{v'\to c}$ \label{line:boxplus}
    }
    

    \For{$v\in \vv_\vtn, \mat W \in \{\mat X, \mat Y, \mat Z\}$}{
            $\Gamma_{v}^\mat{W} \leftarrow \Lambda^{\mat W}_v(t) + \sum_{\substack{c \in \ee(v)\\\langle\mat W, \mat H_{cv}\rangle = 1}} \Delta_{c \to v}$
        }

    \For{$v \in \vv_\vtn, c\in \ee(v), \Wsf\in \{\Xsf, \Ysf, \Zsf\}$}{
            $\Gamma_{v\to c}^\mat{W} \leftarrow \Gamma^{\mat W}_{v} - \langle\mat W, \mat H_{cv}\rangle \Delta_{c \to v}$
        }
        
        \For{$i \in \{0,1\}, c \in {\vv_{\textnormal{c},i}}, v \in \ee(c)$}{
        $\Gamma_{v\to c} \leftarrow
            \log (1 + {\mathrm e}^{-\Gamma_{v\to c}^{\Wsf_i}})+
            \min(\Gamma_{v\to c}^\Ysf, \Gamma_{v \to c}^{\Wsf_{\Bar{\imath}}}) 
            - \log \left(1 + {\mathrm e}^{-\left|\Gamma_{v \to c}^\Ysf - \Gamma_{v \to            c}^{\Wsf_{\Bar{\imath}}}\right|}\right)$
        }

    Let $\hat{\bm{\mat E}}(t) \in \{\Isf, \Xsf, \Ysf, \Zsf\}^{n}$ be the vector where $\hat{{\mat E}}_v(t) = \Isf$ if $\Gamma_v^\Wsf > 0$ for all $\Wsf\in \{\Xsf, \Ysf, \Zsf\}$ and $\hat{{\mat E}}_v(t) = \argmin_{\Wsf\in \{\Xsf, \Ysf, \Zsf\}} \Gamma_v^\Wsf$ otherwise.

    \If{$\langle \hat{\bm{\mat E}}(t) , \mat H_{i,:}\rangle = s_i$ \textnormal{for all rows} $\mat H_{i,:} \subseteq \mat H$}{

        $\omega_r \leftarrow \omega(\hat{\bm{\mat E}}(t)) = \sum_{v\in \vv_\textnormal{v}, \hat{\bm{\mat E}}(t)_v = \Wsf \neq \Isf} \Lambda_v^\Wsf$

        \If{$\omega_r < \omega_{\hat{\bm{\mat E}}}$}{
        $\hat{\bm{\mat E}} \leftarrow \hat{\bm{\mat E}}(t)$

        $\omega_{\hat{\bm{\mat E}}} \leftarrow \omega_r$
    
    }
    \textbf{break};
    }
    
    }
    \If{$s = S$}{
    \textbf{break};
    }
    }

   \Return{$\hat{\bm{\mat E}}$}
\end{algorithm}

\section{Conclusion and Future Work}

We proposed an improved iterative BP decoder for quantum Tanner codes that  exploits the underlying local code structure by combining check nodes into more powerful generalized checks that are decoded using a MAP decoder as part of the check node processing of each decoding iteration. The proposed decoder significantly outperforms the standard quaternary BP decoder with memory effects, as well as the recently proposed Relay-BP decoder, at the expense of a higher decoding complexity, on the depolarizing channel. With the new decoder, we show that quantum Tanner codes can outperform GB and LP codes of  comparable parameters, which shows that this class of codes can be competitive also in the finite codelength regime. Notably, combining checks into more powerful ones using a proposed greedy algorithm does not give noticeable gains for other classes of qLDPC codes, like LP and HGP codes.
To explain the gain in performance for quantum Tanner codes, we showed how the $4$-cycles of their Tanner graphs are reduced when combining checks, and gave a way to construct quantum Tanner codes where the Tanner graph has a designed girth when combining checks.

As future work, we will leverage the product structure of the local codes as well as investigate suboptimal SISO decoding algorithms for generalized checks to lower the decoder's overall complexity.
It would also be interesting to investigate additional constructions of quantum Tanner codes that satisfy the $2$-TNC condition and to study their performance.
Finally, evaluating the performance under circuit level noise models is an interesting avenue for further research.  

\section*{Acknowledgment}

The authors would like to thank Michele Pacenti for valuable discussions on the decoding of quantum Tanner codes.

\appendices

\section{Relay-BP$_4$} \label{app:relay_bp}
Below in Algorithm~\ref{alg:Relay-BP4}, we give the version of Relay-BP that we have run in our simulations, referred to as Relay-BP$_4$. Relay-BP$_4$ is a quaternary version of \cite[Alg.~1]{Mueller-etal25_1sub}. In addition, we use the conventional box-plus check updating rule in \cref{line:boxplus} instead of the scaled-min check updating rule in \cite[Eq.~(1)]{Mueller-etal25_1sub}.

\section{Example: The $[[45, 7, 4]]$ quantum Tanner code}\label{ex:Tanner-graph-45-7-4}

    Consider the $[[45, 7, 4]]$ quantum Tanner code of \cref{ex:45-7-4-2TNC}.
    On each vertex of $\Gs_0$ and $\Gs_1$, the checks 
    \begin{align*}
        \hha^\perp \otimes \hhb^\perp &= \begin{bmatrix}
            1&0&1&1&0&1&1&0&1 \\ 0&1&1&0&1&1&0&1&1
        \end{bmatrix} \\
        \hha \otimes \hhb &= \begin{bmatrix}
        1&1&1&0&0&0&1&1&1 \\ 0&0&0&1&1&1&1&1&1
    \end{bmatrix}
    \end{align*}
    are placed, respectively. Since $\Gs_0 = \Gs_1 = \set K_{10}$, a graph with $10$ vertices, this means $\Cc_0$ and $\Cc_1$ are defined by $20$ checks each. Since $\hha \otimes \hhb$ contains $6$ $4$-cycles, $\mat H_1$ contains exactly $60$ $4$-cycles by \cref{cor:QT-$4$-cycles-Ci}, and combining 2 checks at a time, we get back the incidence graph of $\set K_{10}$, a graph with no $4$-cycles. The Tanner graph is shown in \cref{fig:45-7-4-TG}, where square check nodes are combined to a larger rectangle two at a time. A standard BP decoder will do message passing on the graph with square check nodes, while the proposed generalized decoder can do message passing on the $4$-cycle free graph with rectangular check nodes instead.

\begin{figure}
    \centering
\begin{tikzpicture}[scale=2.0,auto, swap, node distance=10m,
circ/.style={draw,shape=circle, inner sep=0.3ex, minimum size=3ex}, sq/.style={draw,shape=rectangle, inner sep=0.3ex, minimum size=3ex},
lab/.style={shape=circle, font=\scriptsize}]
    \foreach \i in {0,...,9}
    {
    \draw (0, \i+0.2) node[sq] (cu\i) {$\scriptstyle \i_{0}$};
    \draw (0, \i-0.2) node[sq] (cd\i) {$\scriptstyle \i_{1}$};
    \node[fit=(cu\i)(cd\i), draw] {};
    }
    \foreach \i in {0,...,44}
    {
    \draw (3, \i/4 - 1) node[circ] (r\i) {$\scriptstyle \i$};
    }
    \foreach \x [evaluate=\x as \t using {int(\x - 1)}] in {1,2,3,7,8,9} {
    \draw (cu0) -- (r\t);
    };
    \foreach \x [evaluate=\x as \t using {int(\x - 1)}] in {4,5,6,7,8,9} {
    \draw (cd0) -- (r\t);
    };
    \foreach \x [evaluate=\x as \t using {int(\x - 1)}] in {1,10,11,12,16,17} {
    \draw (cu1) -- (r\t);
    };
    \foreach \x [evaluate=\x as \t using {int(\x - 1)}] in {1,13,14,15,16,17} {
    \draw (cd1) -- (r\t);
    };
    \foreach \x [evaluate=\x as \t using {int(\x - 1)}] in {2,10,18,19,20,24} {
    \draw (cu2) -- (r\t);
    };
    \foreach \x [evaluate=\x as \t using {int(\x - 1)}] in {2,10,21,22,23,24} {
    \draw (cd2) -- (r\t);
    };
    \foreach \x [evaluate=\x as \t using {int(\x - 1)}] in {3,11,18,25,26,27} {
    \draw (cu3) -- (r\t);
    };
    \foreach \x [evaluate=\x as \t using {int(\x - 1)}] in {3,11,18,28,29,30} {
    \draw (cd3) -- (r\t);
    };
    \foreach \x [evaluate=\x as \t using {int(\x - 1)}] in {12,19,25,31,32,33} {
    \draw (cu4) -- (r\t);
    };
    \foreach \x [evaluate=\x as \t using {int(\x - 1)}] in {4,12,19,25,34,35} {
    \draw (cd4) -- (r\t);
    };
    \foreach \x [evaluate=\x as \t using {int(\x - 1)}] in {20,26,31,36,37,38} {
    \draw (cu5) -- (r\t);
    };
    \foreach \x [evaluate=\x as \t using {int(\x - 1)}] in {5,13,20,26,31,39} {
    \draw (cd5) -- (r\t);
    };
    \foreach \x [evaluate=\x as \t using {int(\x - 1)}] in {27,32,36,40,41,42} {
    \draw (cu6) -- (r\t);
    };
    \foreach \x [evaluate=\x as \t using {int(\x - 1)}] in {6,14,21,27,32,36} {
    \draw (cd6) -- (r\t);
    };
    \foreach \x [evaluate=\x as \t using {int(\x - 1)}] in {7,33,37,40,43,44} {
    \draw (cu7) -- (r\t);
    };
    \foreach \x [evaluate=\x as \t using {int(\x - 1)}] in {15,22,28,33,37,40} {
    \draw (cd7) -- (r\t);
    };
    \foreach \x [evaluate=\x as \t using {int(\x - 1)}] in {8,16,38,41,43,45} {
    \draw (cu8) -- (r\t);
    };
    \foreach \x [evaluate=\x as \t using {int(\x - 1)}] in {23,29,34,38,41,43} {
    \draw (cd8) -- (r\t);
    };
    \foreach \x [evaluate=\x as \t using {int(\x - 1)}] in {9,17,24,42,44,45} {
    \draw (cu9) -- (r\t);
    };
    \foreach \x [evaluate=\x as \t using {int(\x - 1)}] in {30,35,39,42,44,45} {
    \draw (cd9) -- (r\t);
    };
\end{tikzpicture}
    \caption{The Tanner graph of the $[[45, 7, 4]]$ quantum Tanner code of \cref{ex:45-7-4-2TNC} and \cref{tab:2-TNC-codes} with check nodes depicted as squares and variable nodes as circles. When grouping two and two check nodes as depicted, the resulting graph has girth six.}
    \label{fig:45-7-4-TG}
\end{figure}


\bibliographystyle{IEEEtran}
\bibliography{../defshort1,../biblioHY,../references}

\end{document}